\documentclass[twocolumn]{IEEEtran} 
\usepackage{epsfig,amsmath,amssymb,epsf,cite, amsthm, scalefnt, subfigure}
\usepackage{graphicx, amsmath, amssymb, algorithm, algorithmic}
\usepackage{color}

\newtheorem{theorem}{Theorem}
\newtheorem{lemma}{Lemma}
\newtheorem{corollary}{Corollary}

\newtheorem{proposition}{Proposition}
\newtheorem{remark}{Remark}
  
\def\cD{{\mathcal{D}}} \def\cE{{\mathcal{E}}} 
\def\cG{{\mathcal{G}}}  
  
 \def\cN{{\mathcal{N}}} 
  
\def\cS{{\mathcal{S}}}  
\def\cV{{\mathcal{V}}}

\def\bSigma{{\pmb{\Sigma}}} \def\bSigma{{\pmb{\Sigma}}}

\def\bTheta{{\pmb{\Theta}}} \def\btheta{{\pmb{\theta}}}
  \def\b0{{\pmb{0}}}

 \def\bb{{\mathbf{b}}} 
  
\def\bg{{\mathbf{g}}}  
  
 \def\bn{{\mathbf{n}}} 
 \def\bq{{\mathbf{q}}} \def\br{{\mathbf{r}}}
\def\bs{{\mathbf{s}}}  \def\bu{{\mathbf{u}}}
\def\bv{{\mathbf{v}}} \def\bw{{\mathbf{w}}} \def\bx{{\mathbf{x}}}
\def\by{{\mathbf{y}}} \def\bz{{\mathbf{z}}}

 \def\bB{{\mathbf{B}}} \def\bC{{\mathbf{C}}}
  \def\bF{{\mathbf{F}}}
 \def\bH{{\mathbf{H}}} \def\bI{{\mathbf{I}}}
 \def\bK{{\mathbf{K}}} 
  
\def\bP{{\mathbf{P}}}

\def\sc{{\boldsymbol{\mathsf{c}}}} 
 \def\sf{{\boldsymbol{\mathsf{f}}}}

 \def\sl{{\boldsymbol{\mathsf{l}}}}

\newcommand{\be}{\begin{equation}}
\newcommand{\ee}{\end{equation}}

\begin{document}

\title{Communicating over Filter-and-Forward Relay Networks with Channel
Output Feedback \thanks{Copyright (c) 2015 IEEE. Personal use of this
material is permitted. However, permission to use this material for any other purposes
must be obtained from the IEEE by sending a request to
pubs-permissions@ieee.org. 
The material in this paper was presented in part at
the \textit{50th Annual Allerton Conference on Communication, Control, and
Computing}, Monticello, IL, October 2012~\cite{AgLo12}. D. J. Love and V.
Balakrishnan are with the School of Electrical and Computer Engineering, Purdue
University, West Lafayette, IN 47907 USA (e-mail: djlove@ecn.purdue.edu;
ragu@ecn.purdue.edu). M. Agrawal is with WorldQuant Research (India) Private
Limited, Mumbai, MH 400076 INDIA (e-mail: agrawal.mayur@gmail.com).}}
\author{Mayur Agrawal, David J. Love, and Venkataramanan Balakrishnan}

\maketitle
{
\begin{abstract}
Relay networks aid in increasing the rate of communication  {from
source to destination}. {However, the capacity of even a
three-terminal relay channel is an open problem.} {In this work,
we propose a new lower bound for the capacity of the three-terminal relay
channel with destination-to-source feedback in the presence of noise with
memory.
Our lower bound improves on the existing bounds in the literature. } We then
extend our lower bound to general relay network configurations using an
arbitrary number of filter-and-forward relay nodes. {Such network
configurations are common in many multi-hop communication systems where the
intermediate nodes can only perform minimal processing due to limited
computational power. Simulation results show that significant improvements in
the achievable rate can be obtained through our approach}. We next derive a
coding strategy~(optimized using post processed signal-to-noise ratio as a
criterion) for the three-terminal relay channel with noisy channel output
feedback for two transmissions. This coding scheme can be used in conjunction
with open-loop codes for applications like automatic repeat request~(ARQ) or
hybrid-ARQ.
\end{abstract}

}

\section{Introduction}\label{sec:intro}
The use of feedback in {communication systems} has the potential
to greatly simplify encoding and decoding processes~\cite{Sh58, ScKa66, Sc66,
CoPo89}. In the seminal work done by Schalkwijk and Kailath~(SK)
in~\cite{ScKa66, Sc66}, they proposed a capacity-achieving linear coding scheme
for a point-to-point link with noiseless channel output feedback. Despite its
simplicity, the SK scheme achieves doubly exponential decay with blocklength in
the probability of error for additive white Gaussian noise~(AWGN) channels.
Variations of the SK scheme have been shown to achieve the capacity and
{provide increased reliability} for a certain class of colored
channels as well~\cite{Bu69,Ki06, Ki10}. { In \cite{SaAg13}, the
capacity of a two-user interference channel under different channel output
feedback architectures has been studied.} {Also, it has been} shown in
\cite{PoPo11} that feedback can lead to a reduction in transmit power for the same forward rate constraint.

The proliferation of wireless systems has spurred much research on the
use of relays~(e.g., \cite{PaWa04, MuVi07, GeMu08}). 
{Recently, there has been an effort to use relays for general
multicast communication problems~\cite{GuYe13}.} Relay channels could potentially
become the fundamental building blocks for wireless systems in the future. The
concept of the three-terminal relay channel was first introduced in~\cite{Me71}.
{Since then, many coding schemes have been proposed to {exploit} the
advantages offered by the addition of a relay node over a point-to-point
link~\cite{CoGa79,ChMo07,
GaMo06}.  Even with the same total transmission power budget as a
point-to-point link, the presence of a relay node~(along with an intelligent
power allocation scheme between the source node and the relay node) can help
increase the capacity under a Gaussian channel assumption~\cite{ZaOe10, KuLa09,
ZaOe14, KuLi14}.}

{The Gaussian relay channel {utilizing} linear time
invariant relay filtering has been explored earlier in the literature~\cite{ChYo12,
GoYa13}. With no feedback, it has been shown in \cite{ChYo12}, that an
amplify-and-forward~(AF) scheme is the optimal scheme among the class of one-tap
filters. For inter-symbol interference channels, the algorithm in \cite{ChYo12}
shows significant improvements in achievable rate by jointly designing source and relay filters. The work in \cite{GoYa13} explores the design of any arbitrary
causal linear relay precoders. Considering the problem of the joint optimization
of input covariance matrix and relay precoder, they demonstrate that a
subdiagonal precoder is sufficient to achieve the maximum rate of the linear
precoded relays. At low source transmit powers, effectively, the system reduces
to a half-duplex relay, where the relay retransmits information only every alternate time slot.}


{\color{black} The Gaussian parallel relay channel was studied in \cite{ScGa00} in
which the authors proposed amplify-and-forward and decode-and-forward
strategies. The work in \cite{NiDi13} improved on the AF
scheme by proposing a bursty AF coding strategy.
Generalization of bursty AF by linear time-varying relaying is considered in
\cite{XuKi14}. Using the techniques developed in \cite{Ki10}, the authors in
\cite{XuKi14} developed an expression for the maximum achievable rate for
Gaussian parallel relay channels with only linear operations at the relay
nodes. They demonstrate that an optimal coding strategy is
time-sharing among four different AF relaying schemes.}


{The role of feedback in relay channels was initially explored in
\cite{CoGa79}. Under the feedback links available from the destination to both
the source and the relay and from relay to the source (referred to as complete
feedback), it was shown that the cut-set upper bound is achievable via
block-Markov superposition encoding schemes. In the event of the feedback link
available only from the destination to the relay node~(partial feedback),
the relay channel turns into a physically degraded one; thus again
allowing the achievability of cut-set upper bound. {The authors of \cite{KuLa11}} proposed a simple SK-type linear coding scheme to achieve rates very close to the capacity for the above
partial feedback setting. However the feedback capacity under other
partial feedback scenarios is still an open area of
research~\cite{GaBr06}.}

In this work, we look at the communication between a source and a destination
over a relay network with a {partial feedback} link available from
the destination to the source in the form of channel output feedback. The relay
nodes in the network can filter-and-forward the received signal to the next node. For the scenario involving one
{filter-and-forward} relay node in the network~(i.e., a
three-terminal relay channel) with all noises modeled as additive white Gaussian processes, a lower bound on the feedback
capacity has been proposed in \cite{GaBr06,BrWi09}.
We improve on this lower bound by employing a time-invariant finite impulse
response~(FIR) filter at the relay node {which builds on} the
recent success in characterizing the capacity of the stationary Gaussian
channels with channel output feedback for point-to-point links~\cite{Ki10}.
{While such filter-and-forward relays were unsuccessful in
improving the rates for the Gaussian relay channel without
feedback~\cite{ChYo12}, these relays prove quite successful in increasing the achievable rate for the
relay channel with a partial feedback link present between the destination and
the source.} The fundamental observation that we make is that a relay node using an
FIR filter~(without decoding the source message) can be viewed as a \emph {virtual point-to-point} link with colored noise in the feedforward part.
In fact we show numerically that in some two-tap filter cases our
lower bound capacity can be twice as high as the point-to-point communication
link capacity.

Our approach to deriving a lower bound on the feedback capacity enables us to
extend the bound to any stationary auto-regressive moving
average~(ARMA) noise process. In the process, we also suggest an alternate, but
more concise, derivation of the lower bound proposed in~\cite{BrWi09}. We then
extend the lower bound to more general relay network configurations, in
particular to the ones involving multiple amplify-and-forward relays in
parallel and series configurations.

One of the major advantages of the proposed lower bound is that it is achievable
by a generalization of the SK scheme~\cite{Ki10}, implying a very low
computational requirement at all the nodes involved, {i.e., the
source, relays, and destination}. Additionally the SK-type scheme results
in doubly exponential decay in the probability of error as a function of the blocklength used for the
transmission of the source message in the absence of feedback noise.
{Note that our proposed scheme is very different from the
interlacing structure of the precoder at the relay node in \cite{GoYa13}. In an
SK-like scheme, the relay transmits message in every time slot, instead of
working in a half duplex mode as proposed in \cite{GoYa13} for the no feedback
case.}

Moving forward, we extend the development of the three-terminal relay channel
with ideal channel output feedback to the one that involves noise in the
feedback link. {Noisy feedback is a more realistic model and has
been discussed for point-to-point channels with Gaussian noise~\cite{ZaDa11} and
quantization noise~\cite{WiMa14, WiFa14}}. While a simple expression for an
arbitrary blocklength transmission seems intractable, we develop a coding
strategy for the case of two channel uses in the presence of additive white
Gaussian noise.
{The proposed strategy is optimized to increase the post-processed
signal-to-noise-ratio~(SNR) at the destination.} We then analyze the impact that
source-to-relay and source-to-destination noise have on the overall network
performance. When specialized to the point-to-point link~(by turning off the
relay node), we recover the result in \cite{Bu69} for noisy feedback. For
practical implementation, the proposed coding scheme can be used as an inner
code in a concatenated fashion as outlined in \cite{ZaDa11}.
 {While the development in \cite{ZaDa11} is for a point-to-point
link, our scheme can be used for a link with relay nodes, thereby finding ready
usage in applications like automatic repeat request~(ARQ) and hybrid-ARQ. Our formulation also differs
from the one in \cite{ZaDa11} by forcing equal power constraint on each
channel usage as opposed to a sum power constraint over the two channel uses.}

The remainder of the paper is organized as follows. Section \ref{sec:sm}
describes the mathematical formulation for the system, assuming a network with
filter-and-forward relay nodes. {This is then specialized to a
three-terminal relay problem, followed by a mathematical reformulation that sets
up the framework for optimization.} In Section \ref{sec:NCOF}, we present a
lower bound on the feedback capacity for any arbitrary colored Gaussian noise
three-terminal relay channel. We next discuss some illustrative cases of the
general lower bound on the feedback capacity to highlight the advantages offered
by our proposed lower bound.
Section \ref{sec:NCOF_N_2} analyzes the three-terminal node for the case of
noisy channel output feedback for the blocklength size of two. The section also
analyzes the extreme cases possible for the source-to-relay and
source-to-destination noise process. We conclude with a discussion in Section
\ref{sec:conc}.

\subsection*{Notation:}

{Vectors~(matrices) are represented by lower~(upper) boldface
letters, and scalars are represented by lower italicized letters}. The
operators $(\cdot)^T, \textrm{tr}(\cdot)$, and $\lVert \cdot \rVert$ denote the
transpose, trace, and Frobenius norm of a matrix/vector, respectively. The expectation of a
random variable or matrix/vector is denoted by $E[\cdot]$. The boldface letter $\bI$
represents the $N\times N$ identity matrix, and $\cN(0,1)$ denotes the
distribution of a standard normal Gaussian random variable with zero mean and
unit variance.

We now define some representations that will be used frequently in the
subsequent presentation.

\textit{Definition 1:} A random process $\{\widetilde{z}[k]\}_{k = 1}^{\infty}$
is said to be an ARMA$(p,q)$ process if $\widetilde{z}[k]$ evolves as 
\begin{equation}\label{def_arma}
\sum_{j = 0}^{p}\beta_j \widetilde{z}[{k - j}] = \sum_{j = 0}^{q}\alpha_j
\epsilon[{k - j}],
\end{equation}
where each of the $\epsilon[i]$ is an independent and identically
distributed~(i.i.d.) Gaussian random variable with $\cN(0,1), \beta_j \in
\mathbb{R}$ for all $j$,  $\alpha_j \in \mathbb{R}$ for all $j$, and $\beta_0 =
1$.

In the event that $q = 0$ in (\ref{def_arma}), we call the resulting random
process AR($p$). Similarly if $p = 0$ in (\ref{def_arma}), the random process
is called MA($q$) process.

\textit{Definition 2:} An ARMA$(p,q)$ process in (\ref{def_arma}) can be
represented alternatively by defining the delay operator $D$ where
$D^j\widetilde{z}[{k}] = \widetilde{z}[{k - j}].$
Hence, (\ref{def_arma}) can be represented as

\begin{equation}\label{bdef_arma}
G(D)\widetilde{z}[{k}] = F(D) \epsilon[k],
\end{equation}
where $G(D)$ and $F(D)$ are the polynomials given by
\begin{equation}\label{arma_poly}
G(D) = \sum_{j = 0}^p \beta_jD^j, \quad F(D) = \sum_{j = 0}^q \alpha_jD^j.
\end{equation}

An ARMA~$(p,q)$ process $\{\widetilde{z}[k]\}_{k =
1}^{\infty}$ is \textit{said to be stable} if the zeros of $G(D)$ as defined in
(\ref{arma_poly}) lie strictly outside the unit circle.


\textit{Definition 3:} An ARMA~$(p,q)$ process $\{\widetilde{z}[k]\}_{k =
1}^{\infty}$ can be represented using a state space model as \cite{YaKa04}
\begin{subequations}\label{ss_n}
\begin{align}
\bb[{k + 1}] & = \bP\bb[{k}] + \bq\epsilon[k]\\
\widetilde{z}[k] & = \alpha_0\br^T\bb[k] + \alpha_0\epsilon[k], 
\end{align}
\end{subequations}
where $\bb[k] \in \mathbb{R}^{d
\times 1}$ with $d = \max(p,q),$ and the matrices $\bP, \bq$, and $\br$ are
given by
\begin{align*}
\bP & = \left[\begin{array}{cccc}
-\beta_1 & -\beta_2 & \ldots & -\beta_d\\
1 & 0 & \ldots & 0\\
0 & 1 & \ldots & 0\\
\vdots & \vdots & \ddots & \vdots\\
0 & 0 & \ldots & 0
\end{array}\right], \quad \bq = \left[\begin{array}{c}
1\\0\\ 0\\\vdots\\0\end{array}\right],\\
\br &  = \left[\left(\frac{\alpha_1}{\alpha_0} - \beta_1\right), 
\left(\frac{\alpha_2}{\alpha_0} - \beta_2\right),\ldots,
\left(\frac{\alpha_d}{\alpha_0} - \beta_d\right)\right]^T.
\end{align*}


\section{Mathematical Formulation}\label{sec:sm}
 
\subsection{Network Model}
Consider a real-valued discrete time model as shown in
Figure~\ref{fig:nm}.
In this setup, we have a source node $\cS$, a destination node $\cD$, and a relay
network denoted by the directed acyclic graph, $\cG = (\cV,\cE)$. The set $\cV$
contains all the $|\cV|$ nodes in the network, i.e., $\cV = \{\mathsf{v}_1,
\mathsf{v}_2, \ldots, \mathsf{v}_{|\cV|}\}$, while $\cE$ contains the directed
edges of all the connected node pairs, i.e., $\cE = \{(\mathsf{v}_1,
\mathsf{v}_2), (\mathsf{v}_1, \mathsf{v}_6), (\mathsf{v}_5,
\mathsf{v}_1), \ldots \}$. Note that the pair $(\mathsf{v}_5,
\mathsf{v}_1)$ denotes a directed edge from the relay node $\mathsf{v}_5$ to the
relay node $\mathsf{v}_1$.

\begin{figure}
\centering
\includegraphics[scale = 0.40]{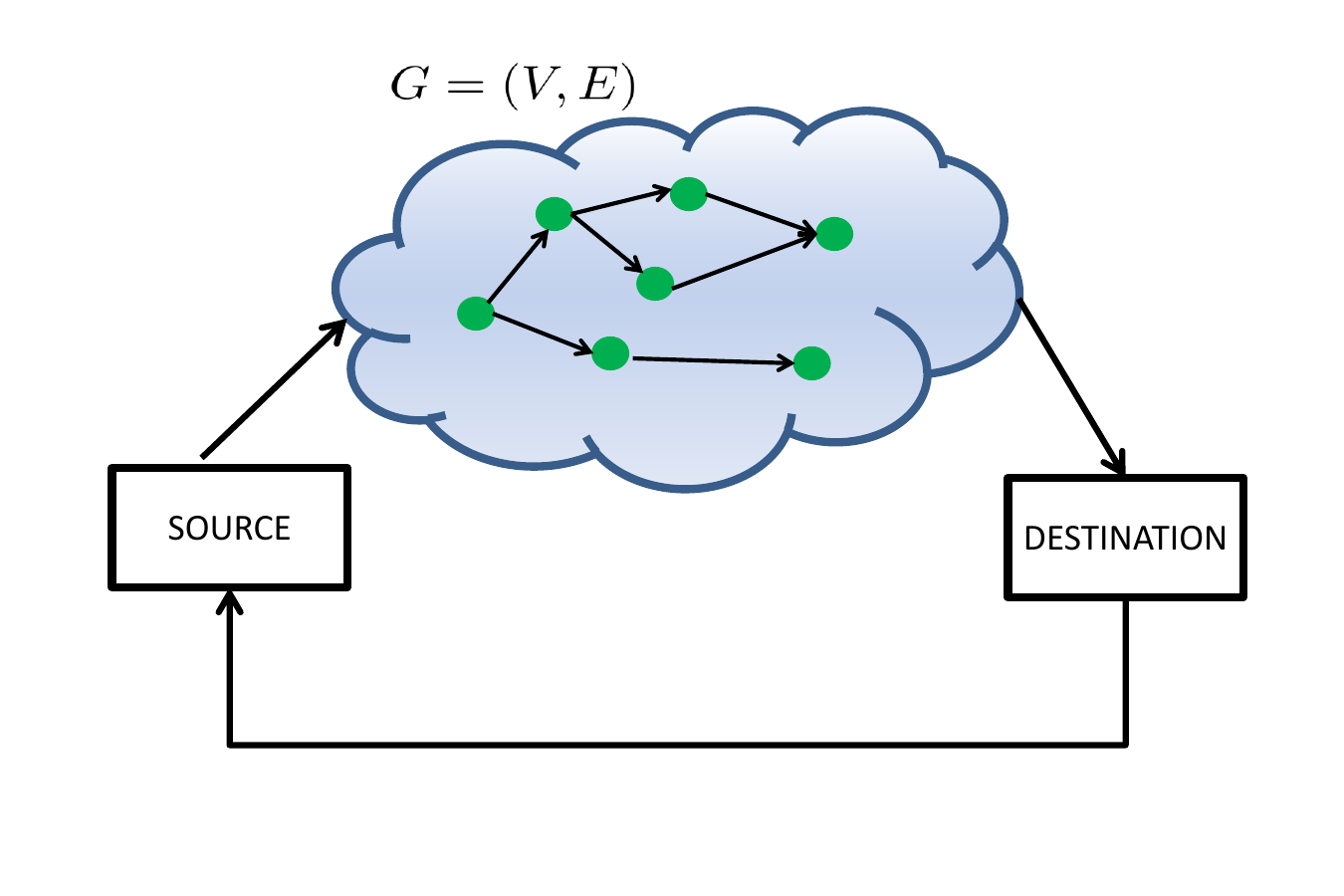}
\caption{System model for network with filter-and-forward relay nodes.}
\label{fig:nm}
\end{figure} 

The signals transmitted by the source $\cS$ and received by the destination
$\cD$ at time instant $k$ are denoted by $x[k]$ and $y[k]$ respectively. The
source communicates the message symbol ${\btheta}$
to the destination $\cD$ over $N + L$ channel uses. Whereas, the first $N$
channel outputs correspond to channel uses involving a message being sent, the
last $L$ outputs are due to the memory of the network.

The symbol $\btheta$ is
drawn uniformly from a symbol constellation of $M$ symbols denoted by $\bTheta =
\{\btheta_1, \ldots, \btheta_M\}$. Without loss of generality~(WLOG), we put a
norm constraint on the symbol constellation, $E[\lVert \btheta\rVert^2] = 1$. The
average power constraint at the source is specified by
\begin{equation}
\frac{1}{N + L}\sum_{k = 1}^{N + L}E[x^2[k]] \leq \rho,
\end{equation}
with $x[N+1] = \cdots = x[N + L] = 0$.

In our network model, we assume that the relay nodes do not have the
computational resources to decode the information transmitted to them by either
the source or the other relay node. The node $\mathsf{v}_i$ can only linearly
combine the previous $L_i$ received signals. In other words, the relay node
implements an $L_i$-tap time invariant FIR filter whose output at the time
instant $k$ is given by
\begin{equation}
v_i[k] = \sum_{\ell = 1}^{L_i}h_i[\ell] u_i[k - \ell],  
\end{equation}
where $\{h_i[\ell]\}_{\ell = 1}^{L_i}$ are the coefficients of the FIR filter
at the relay, $u_i[k]$ is the input to the relay node $\mathsf{v}_i$ at time
$k$, and $v_i[k]$ is the output at time instant $k$. Also, we specify the
additional power constraint at the relay node $\mathsf{v}_i$ by
\begin{equation}\label{eq:cr}
\frac{1}{N + L}\sum_{k = 1}^{N + L}E[v_i^2[k]] \leq \gamma_i \rho, \quad
\gamma_i > 0.
\end{equation} 

Note that we define the memory of the network $L$ by the following expression
\begin{equation}
L = \max_{\textrm {All paths~} \cS \rightarrow \cD} \left(\sum_{i: v_i \in
{~\textrm{a path}~} \cS \rightarrow \cD} L_i\right).
\end{equation}

\noindent In other words, $L$ is the total delay of the impulse response of
the system composed of all the relay nodes (i.e., source to destination).
{The above definition of $L$ allows us to view the
complete network as an effective relay node. This characterization will be
especially useful in the next sub-section where we investigate our original
network problem using an effective three-terminal relay problem.}

Furthermore, the input to the relay node $\mathsf{v}_i$ at time instant $k$ is
given by
\begin{equation}
u_i[k] = \sum_{j:(\mathsf{v}_j,\mathsf{v}_i )\in \cE}v_j[k] + w_i[k],
\end{equation}
where $w_i[k]$ is an additive ARMA~($p_i, q_i$) Gaussian noise with $\cN
(0,\sigma_i^2)$.
It is further assumed that there is a unit delay noiseless feedback link
available from the destination to the source. In other words, when designing $x[k]$, the
source has access to all the previous outputs $\{y[1], \ldots, y[k - 1]\}$. 

Given this arbitrary network of relays with a feedback we consider the following
questions:
\begin{enumerate}
  \item  How do the source and destination perform encoding and decoding to
  exploit the feedback link available between them?
  \item For a fixed network relay (i.e., the FIR filter at each node is
  predetermined), what is the best possible performance that can be achieved?
  \item If we are given the flexibility to even design the filters at each
  node, how can we improve the achievable rate for the network?
\end{enumerate}
  
To address the first question, we make use of linear coding at both the source
and the destination as envisioned by the SK scheme. Furthermore, it will be
shown that for a fixed network relay, we can replace the complete relay network by an
equivalent FIR filter node. In addition we will demonstrate that in the event
that we have the full flexibility to design coefficients at the relay node, it
may not always be optimal to consume the total power available at the
relay.

\begin{figure*}
\centering
\includegraphics[scale = 0.50]{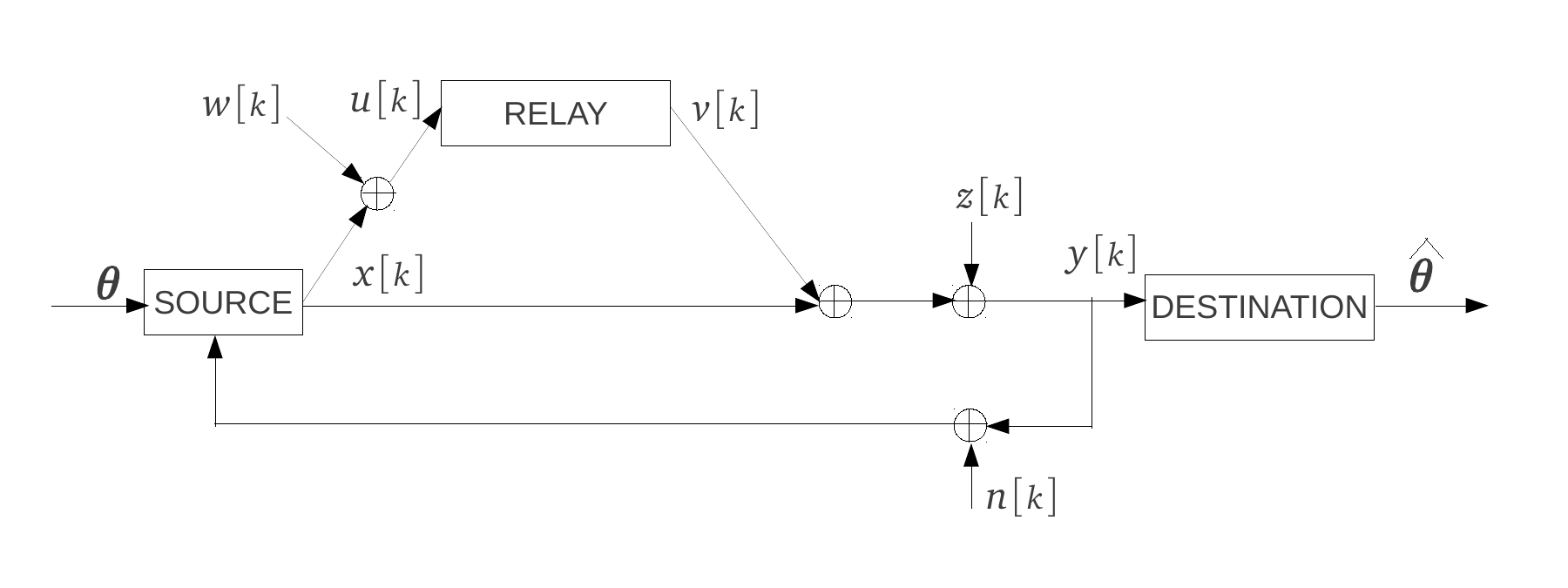} 
\caption{System model for a three-terminal relay channel with channel output
feedback link from the destination to the source.}
\label{fig:sm}
\end{figure*}

\subsection{Three-Terminal Relay}
In this subsection, we consider the case when the FIR filters at the relay
nodes are fixed. {As only fixed linear operations are performed
at each network node, the complete network can be viewed as an effective
single relay node with FIR filter $\{h[\ell]\}_{\ell = 1}^L$.}
Therefore, we begin by exploring the problem of designing a coding scheme for
a three-terminal relay channel. Consider a real discrete-time
three-terminal relay channel with a source node, relay node, and destination
node as depicted in Figure \ref{fig:sm}. The signal received by the destination at the time instant
$k$ is given by
\begin{equation}
y[k] = x[k] + v[k] + z[k], \quad k = 1, \ldots, N + L,
\end{equation}
where $x[k]$ and $v[k]$ are the $k^{th}$ transmitted signals from the source and
the relay, respectively, $z[k]$ is the additive Gaussian noise with
distribution $\cN (0,1)$, and $N + L$ denotes the blocklength for the
transmission of the message symbol ${\btheta}$. Furthermore, the received signal
$u[k]$ at the relay node is given by
\begin{equation}
u[k] = x[k] + w[k], \quad k = 1, \ldots, N + L,
\end{equation}
where $w[k]$ is distributed as $\cN(0, \sigma_w^2)$ where the relay node is
assumed to have $L$ taps.

\subsection{Problem Reformulation}
We can express the received signal $y[k]$ at the destination in terms of the
transmitted signal $\{x[i]\}_{i = 1}^{k}$ and the noise processes $\{w[i]\}_{i =
1}^{k}$ and $\{z[i]\}_{i = 1}^{k}$ as
\begin{align*}
y[k] & = x[k] + v[k] + z[k]\\
& = x[k] + \sum_{i = 1}^{L}h[i]x[k - i] + \sum_{i = 1}^{L}h[i]w[k - i]
+ z[k]. \end{align*}
Neglecting the last $L$ channel uses that are performed to flush the memory
from the network, we can express the first $N$ signal received at the
destination in the vector form as
\begin{align}
\nonumber \!\!\!\!\left[ \begin{array}{c}
y[1] \\
\vdots \\ \vdots \\
\vdots\\
y[N]
\end{array}\right] = & \left[ \begin{array}{cccccc}
1 & 0 & 0 & \ldots & \ldots & \ldots\\
h[1] & 1 & 0 & \ldots & \ldots & \ldots\\
\vdots & \vdots & \ddots & \ddots & \vdots & \vdots \\
\vdots & \vdots & \ddots & \ddots & \vdots & \vdots \\
0 & \ldots & h[{L}] & h[{L - 1}] & \ldots & 1\\
\end{array}\right]
\left[ \begin{array}{c}
x[1] \\
\vdots \\ \vdots \\
\vdots\\
x[N]
\end{array}\right]  \\
\nonumber&
+\underbrace{\left[ \begin{array}{ccccc}
0 & 0 & 0 & \ldots & \ldots\\
h[1] & 0 & 0 & \ldots & \ldots\\
\vdots & \vdots & \ddots & \vdots & \vdots \\
\vdots & \vdots & \ddots & \vdots & \vdots \\
\ldots & h[{L}] & h[{L - 1}] & \ldots & 0\\
\end{array}\right]}_{\bH}
\left[ \begin{array}{c}
w[1] \\
\vdots \\ \vdots \\
\vdots\\
w[N]
\end{array}\right] \\
&+ \left[\begin{array}{c}
z[1] \\
\vdots \\ \vdots \\
\vdots\\
z[N] \label{def_H}
\end{array}\right].
\end{align}
In other words,
\begin{equation}\label{eq:me}
\by = (\bI + \bH) \bx + \bH\bw + \bz,
\end{equation}
where $\by = \left[y[1], y[2], \ldots, y[N]\right]^T,$ and the definition of
other vectors similarly follows. Hence the input and the output signal vectors at the
relay are given by
\begin{equation}\label{eq:ro}
\bu  = \bx + \bw, \quad \bv = \bH\bu = \bH\left(\bx + \bw\right).
\end{equation}
{Because the lower triangular matrix $\left(\bI + \bH \right)$ has ones along
its principal diagonal, it is invertible}. Furthermore, the inverse of a
lower triangular matrix is also lower triangular~\cite{GoVa96}. This allows us
to perform \textit{causal} linear processing in (\ref{eq:me}) to obtain
\begin{equation}
\widetilde{\by} = \bx + (\bI + \bH)^{-1}\bH\bw + \left(\bI + \bH\right)^{-1}\bz,
\end{equation}
where  $\widetilde{\by} = \left(\bI + \bH\right)^{-1}\by$. By causal processing,
we mean that the entry $\widetilde{y}[k]$ is only a deterministic function of the values
$\{y[1], \ldots, y[k]\}$. Let us define the effective noise as
\begin{equation}\label{def_n}
\widetilde{\bz} \stackrel{\Delta}{=} (\bI + \bH)^{-1}\bH\bw + \left(\bI +
\bH\right)^{-1}\bz.
\end{equation}
The covariance of the noise vector $\widetilde{\bz}$ is given by
\begin{equation}\label{pr:covar_n}
\bK_{\widetilde{z}} = (\bI + \bH)^{-1}\bH\bK_w \bH^T(\bI + \bH^T)^{-1} +
(\bI + \bH)^{-1}\bK_z(\bI + \bH^T)^{-1},
\end{equation}
where $\bK_w = E[\bw \bw^T]$ and $\bK_z = E[\bz \bz^T]$.
With the effective noise $\widetilde{\bz}$, the processed signal at the
destination can be written as
\begin{equation}\label{eq:sm1}
\widetilde{\by} = \bx + \widetilde{\bz},
\end{equation}
where $\bx$ is the signal of interest and $\widetilde{\bz}$ is the additive
colored Gaussian noise. The signal $\bx$ has the power constraint
\begin{equation}\label{eq:sm2}
\frac{1}{N + L}E[\bx^T \bx] \leq \rho.
\end{equation}

In the event that we also have the flexibility in designing the relay node
coefficients, we need to satisfy an additional power constraint in
(\ref{eq:cr}) at the relay node given by
\begin{equation}\label{eq:sm3}
\frac{1}{N + L}\textrm{tr}\left(E\left[\bv \bv^T\right]\right) \leq \gamma \rho.
\end{equation}
Substituting the value of $\bv$ from (\ref{eq:ro}), the constraint can be
re-written in terms of the input vector $\bx$ as
\begin{equation}\label{eq:sm4}
\frac{1}{N + L}\textrm{tr}\left( \bH E\left[(\bx + \bw)(\bx +
\bw)^T\right]\bH^T\right) \leq \gamma \rho.
\end{equation}

Note that the presence of the feedback link with a unit delay ensures that the
source has access to the side-information from the destination. In
particular, we assume that the side-information is a noise corrupted version
of the received signal $y[k]$. This additional side information means that when
designing $x[k]$, the source has access to the previous
corrupted outputs $\{y[i] + n[i]\}_{i=1}^{k-1}$, where $n[i]$ has distribution
$\cN(0,\sigma_n^2)$.
Also, the noise processes $\{w[k]\}_{k = 1}^{N}$, $\{z[k]\}_{k = 1}^{N}$,
and $\{n[k]\}_{k = 1}^{N}$ are assumed to be independent of each
other. 


\section{Noiseless Channel Output Feedback for a Three-terminal
Relay}\label{sec:NCOF} 
In this section, we look at the ideal case of noiseless
channel output feedback, i.e., $\sigma_n^2 = 0$. With the noiseless channel output information 
at the source, the source has perfect knowledge of the estimate of the original
message $\btheta$ at the destination.
\subsection{$(N, L)$-block Feedback Capacity Optimization}

The formulation in (\ref{eq:sm1}), (\ref{eq:sm2}), and (\ref{eq:sm3}) of the
original three-terminal relay channel is similar to the point-to-point
communication link with the feedback link as discussed in~\cite{CoPo89}, but
with an additional power constraint at the relay given by (\ref{eq:sm3}). Using
the generalized notion of capacity as described in \cite{CoPo89}, we define
\begin{equation}{\label{eq:nlcapacity1}}
C_{\rm{FB}, N, L} = \sup_{\bK_x, \bH}\frac{1}{2(N + L)}\log\frac{\det
\bK_{x + \widetilde{z}}}{\det\bK_{\widetilde{z}}},
\end{equation}
where the above maximization needs to satisfy the power constraints as outlined
in  (\ref{eq:sm2}), and (\ref{eq:sm3}). The above quantity can be
thought of as the capacity of the channel if it is used for a block
of length $N + L$. 

Note that the notion of $(N, L)-$block feedback capacity~($C_{\rm{FB}, N, L}$)
in (\ref{eq:nlcapacity1}) for the original three-terminal relay channel holds for
any time varying noise process. However the limiting feedback capacity for any
arbitrary noise process may not exist. In the event that the limit exists,
we define the feedback capacity as
\begin{equation}\label{eq:CfbOpt}
C_{\rm{FB}, L} = \lim_{N \rightarrow \infty}C_{\rm{FB}, N, L}.
\end{equation} 

\begin{lemma}\label{NLOpt}
The $(N, L)-$block feedback capacity optimization for the $L-$tap three-terminal
relay is given by
\begin{subequations}\label{eq:lb_opt1}
\begin{align}
\!\!\!C_{\rm{FB}, N, L} & = \sup_{\bK_s, \bB,
\bH}\frac{1}{2(N + L)}\log\frac{\det\left(\bK_s + (\bI +
\bB)\bK_{\widetilde{z}}(\bI + \bB)^T\right)}{\det(\bK_{\widetilde{z}})} \\
\nonumber \!\!\!\textrm{such that} \\
& \!\!\!\!\!\quad \textrm{tr}(\bK_s
+ \bB\bK_{\widetilde{z}} \bB^T) \leq (N + L)\rho,\\ 
\nonumber \!\!\!\!\! & \textrm{tr}\left( \bH(\bK_s +
\bB\bK_{\widetilde{z}} \bB^T + \bB\left(\bI + \bH\right)^{-1}\bH\bK_w \right. \\
& \!\!\!\!\!\quad \left. + (\bB\left(\bI + \bH\right)^{-1}\bH\bK_w)^T +
\bK_w)\bH^T\right) \leq \gamma (N + L)\rho.
\end{align}
\end{subequations}
where the maximization is performed over all positive semidefinite
symmetric matrices $\bK_s$, all strictly lower triangular matrices $\bB$, and
all strictly lower triangular $L-$banded Toeplitz matrices $\bH$ (see
(\ref{def_H})). 
\end{lemma}
\begin{proof}
WLOG assume that we are provided with the set of
filter taps, i.e., $\bH$. In the presence of effective noise $\widetilde{\bz}$
given by (\ref{pr:covar_n}) and a noiseless feedback link, it has been shown in
\cite{CoPo89} that the optimal input signal is given by
\begin{equation}
\bx = \bs + \bB\widetilde{\bz},
\end{equation}
where $\bs$ is a signal vector dependent on just the message
$\btheta$ and $\bB$ is a strictly lower triangular matrix to enforce 
causality at the source. As a result, the received signal $\widetilde{\by}$ can
be written as
\begin{equation}
\widetilde{\by} = \bx + \widetilde{\bz} = \bs + (\bI + \bB)\widetilde{\bz}.
\end{equation}
Clearly, $\bK_{\widetilde{y}} = \bK_s + (\bI +
\bB){\bK}_{\widetilde{z}}(\bI +
\bB)^T$ and $\bK_{x} = \bK_s + \bB{\bK}_{\widetilde{z}}\bB^T$.

The $(N,L)-$ block feedback capacity~\cite{CoPo89} can then be expressed
as
\begin{equation}
C_{\rm{FB}, \bH, N, L} = \sup_{\bK_s, \bB
}\frac{1}{2(N + L)}\log\frac{\det
\bK_{\widetilde{y}}}{\det\bK_{\widetilde{z}}}.
\end{equation}
Substituting the value of $\bK_{x}$ in (\ref{eq:sm2}) and (\ref{eq:sm3}), we
immediately get the result of the lemma.
\end{proof}

In the following proposition, we show that for a given set of filter
taps~(i.e., $\bH$ is fixed), the above $(N,L)- $block feedback capacity
optimization for the three-terminal relay problem can be cast as a convex
optimization problem, thereby leading to numerically computable solutions.

\begin{proposition}\label{lemma_convex}
For any given FIR filter at the relay, $\{h[\ell]\}_{\ell = 1}^{L}$, the
optimization in (\ref{eq:lb_opt1}) is convex.
\end{proposition}
\begin{proof}
For any given matrix $\bH$, it is obvious that the covariance of the noise
vector $\widetilde{\bz}, \bK_{\widetilde{z}}$ in (\ref{pr:covar_n}) is constant.
Introducing the new variable $\bK_y = \bK_s + (\bI + \bB)\bK_{\widetilde{z}}(\bI
+ \bB)^T$ as in \cite{VaBo98}, we obtain the new equivalent optimization problem as
\begin{subequations}\label{pr:opt2}
\begin{align}
\max_{\bK_y, \bB} & \quad \log \det \bK_y \\
\nonumber \textrm{such that}\\
& \!\!\!\!\!\!\!\!\!\!\!\!\!\!\!\! \textrm{tr}(\bK_y -
\bB\bK_{\widetilde{z}} - \bK_{\widetilde{z}}\bB^T - \bK_{\widetilde{z}}) \leq (N + L - 1)\rho,\\ 
\nonumber & \!\!\!\!\!\!\!\!\!\!\!\!\!\!\!\!\textrm{tr}( \bH(\bK_y -
\bB\bK_{\widetilde{z}} - \bK_{\widetilde{z}}\bB^T - \bK_{\widetilde{z}} + \bB\left(\bI + \bH\right)^{-1}\bH\bK_w \\
& \!\!\!\!\!\!\!\!\!\!\!\!\!\!\!\!+ (\bB\left(\bI + \bH\right)^{-1}\bH\bK_w)^T +
\bK_w)\bH^T) \leq \gamma (N + L - 1)\rho,\\
& \!\!\!\!\!\!\!\!\!\!\!\!\!\!\!\!\left[\begin{array}{cc}
\bK_y & (\bI + \bB)\\
(\bI + \bB)^T & \bK_{\widetilde{z}}^{-1}\end{array}
\right] \succeq \mathbf{0},
\end{align}
\end{subequations}
which is in fact an instance of a convex optimization problem~\cite{BoVa06}.
\end{proof}

\subsection{A Lower Bound on Capacity}
As noted earlier, a limiting capacity expression in (\ref{eq:CfbOpt}) may not
exist for an arbitrary noise process. Therefore to derive a lower bound on the
limiting capacity expression, we focus on the channels corrupted by
\textit{stationary ergodic Gaussian noise}. Recently in \cite{Ki10}, the
limiting capacity of the forward channel with noiseless feedback in the presence
of stationary Gaussian noise for point-to-point links has been derived.
However, before we proceed, the following proposition that links the order of
the effective noise process $\{\widetilde{z}[k]\}_{k = 1}^{\infty}$ with that of
noise processes \{$w[k]\}_{k = 1}^{\infty}$ and \{$z[k]\}_{k = 1}^{\infty}$ will
be helpful in deriving the expression for the lower bound on the problem posed
previously. The importance of the proposition lies in the observation that the
order of the noise process $\{\widetilde{z}[k]\}_{k = 1}^{\infty}$ is
independent of the number of channel uses $N$.

\begin{proposition}
\label{prop1}
If $\{w[k]\}_{k = 1}^{\infty}$ is an ARMA($p_1, q_1$)
process, $\{z[k]\}_{k = 1}^{\infty}$ an ARMA($p_2,
q_2$) process, and the relay has $L$ taps, then the effective noise process
{\color{black}$\{\widetilde{z}[k]\}_{k = \max(p,q)}^{\infty}$} as defined in
(\ref{def_n}) is also an ARMA($p, q$) process with
\begin{align*}
p & \leq L + p_1 + p_2 ,\\
q & \leq \max(L + p_2 + q_1 - 1, p_1 + q_2).
\end{align*}
{\color{black} To ensure the consistency for all the valid values of $k$, we
assume the initial conditions to be $w[k] = z[k] = \widetilde{z}[k] = 0$ for all $k \leq 0$.
Note that the values for $k \leq \max(p,q)$ introduce non-stationarity due to
the initial conditions that define the random processes $\left\{{z}[k]\right\}_{k =
1}^{\infty}, \left\{{w}[k]\right\}_{k = 1}^{\infty}$ and
$\left\{\widetilde{z}[k]\right\}_{k = 1}^{\infty}$.}
\end{proposition}
\begin{proof}
It can be shown that the inverse of banded Toeplitz matrix
$(\bI + \bH)$ is given by
\begin{equation*}
(\bI + \bH)^{-1} = \left[ \begin{array}{ccccc}
1 & 0 & 0 & \ldots & \ldots\\
a_1 & 1 & 0 & \ldots & \ldots\\
a_2 & a_1 & 1 & 0 & \ldots\\
\vdots & \vdots & \ddots & \vdots & \vdots \\
a_{N - 1} & a_{N - 2} & \ldots & \ldots & 1\\
\end{array}\right],
\end{equation*}
where 
\begin{equation}\label{recur_a}
a_k + \sum_{i = 1}^{L}a_{k - i}{h[i]} = 0, \quad k = 1, \ldots, N - 1,
\end{equation} with the initial
conditions given by  $a_0 = 1, a_{-1} = a_{-2} = \cdots = a_{-(L - 1)} = 0$. 

The $k^{th}$ element of the vector $\widetilde{\bz}$ in (\ref{def_n}) is given
by
\begin{align}
\nonumber \widetilde{z}[k] & = -\sum_{i = 1}^{k -1}a_{k - i}w[i] + \sum_{i =
1}^{k}a_{k - i}z[i]\\
\label{eq:def_n2} & = \sum_{i = 1}^{k -1}a_{k - i}(z[i] - w[i]) + z[k].
\end{align}
Using the definition of $a_k$ from (\ref{recur_a}), we have
\begin{align}
\nonumber \widetilde{z}[k] & = \sum_{i = 1}^{k -1}\left(- \sum_{j = 1}^{L}a_{k -
i - j}h[j]\right)\left(z[i] - w[i]\right)  + z[k]\\
\nonumber & = -\sum_{j = 1}^{L} h[j]\left(\sum_{i = 1}^{k - 1}a_{(k - j) -
i}\left(z[i] - w[i]\right)\right) + z[k]\\
\label{eq:recur_b} & {\color{black} = -\sum_{j = 1}^{L} h[j]\left(\sum_{i = 1}^{k
- j}a_{(k - j) - i}\left(z[i] - w[i]\right)\right) + z[k].}
 \end{align}
 {\color{black}However, using (\ref{eq:def_n2}) and the initial condition that
 $a_0 = 1$, we immediately get
 \begin{equation}
 \label{eq:recur_c}
 \sum_{i = 1}^{k - j}a_{(k - j) -
 i}\left(z[i] - w[i]\right) =  \widetilde{z}[k] - w[k - j].
 \end{equation}
Substituting the value in (\ref{eq:recur_c}) into 
(\ref{eq:recur_b}), we get
\begin{equation*} 
\widetilde{z}[k]= -\sum_{j = 1}^{L} h[j] \widetilde{z}[{k - j}] +\sum_{j = 1}^{L}
h[j]w[{k - j}] + z[k].
\end{equation*}}
Therefore,
\begin{equation}\label{eq:recur_n}
\sum_{j = 0}^{L} h[j] \widetilde{z}[{k - j}]  = \sum_{j = 1}^{L} h[j]w[{k -
j}] + z[k].
\end{equation}
\begin{figure*}
\begin{equation}
\label{eq:recur_n2}\sum_{s = 0}^{p_2}\sum_{r = 0}^{p_1}\sum_{j =
0}^{L}\beta_{s}^{(z)}\beta_{r}^{(w)} h[j] \widetilde{z}[{k - r - s - j}] 
= \sum_{s = 0}^{p_2}\sum_{r =
0}^{p_1}\sum_{j = 1}^{L}\beta_{s}^{(z)}\beta_{r}^{(w)} h[j] w[{k - r - s - j}]  + \sum_{s =
0}^{p_2}\sum_{r = 0}^{p_1}\beta_{s}^{(z)}\beta_{r}^{(w)}z[{k - r - s}].
\end{equation}
\begin{equation}\label{eq:recur_lemma}
\sum_{s = 0}^{p_2}\sum_{r = 0}^{p_1}\sum_{j =
0}^{L}\beta_{s}^{(z)}\beta_{r}^{(w)} h[j] \widetilde{z}[{k - r - s - j}]
= \sum_{j = 1}^{L}\sum_{s = 0}^{p_2}\sum_{r =
0}^{q_1}h[j]\beta_s^{(z)} \alpha_{r}^{(w)}\epsilon^{(w)}[{k - r - s -j}] +
\sum_{r = 0}^{p_1}\sum_{s =
0}^{q_2}\beta_r^{(w)}\alpha_{s}^{(z)}\epsilon^{(z)}[{k - r - s}].
\end{equation}
\end{figure*}
Using (\ref{eq:recur_n}) we have (\ref{eq:recur_n2}).


 With the definition of an ARMA$(p,q)$ noise process, we can represent the noise
 processes $\{w[k]\}_{k = 1}^\infty$ and $\{z[k]\}_{k = 1}^\infty$ as
\begin{subequations}\label{arma_wz}
\begin{align}
\label{arma_w} \sum_{i = 0}^{p_1}\beta_{i}^{(w)}w[k - i] & = \sum_{i =
0}^{q_1}\alpha_{i}^{(w)}\epsilon^{(w)}[k - i],\\
\label{arma_z} \sum_{i = 0}^{p_2}\beta_{i}^{(z)}z[k - i] & = \sum_{i =
0}^{q_2}\alpha_{i}^{(z)}\epsilon^{(z)}[k - i],
\end{align}
\end{subequations}
with $\beta_{0}^{(w)} = \beta_{0}^{(z)} = 1.$

However, using (\ref{arma_wz}) we know that,
\begin{align*}
\sum_{r = 0}^{p_1}\sum_{j = 1}^{L}\beta_{r}^{(w)}
h[j]w[{k - r - s - j}] & = \\
& \!\!\!\!\!\!\!\!\!\!\!\!\!\!\!\!\!\!\!\!\!\!\!\! \sum_{j = 1}^{L}\sum_{r =
0}^{q_1}h[j]\alpha_{r}^{(w)}\epsilon^{(w)}[{k - r - s -j}],
\end{align*}
\begin{equation*}
\!\!\!\!\!\!\!\sum_{s =
0}^{p_2}\sum_{r = 0}^{p_1}\beta_{s}^{(z)}\beta_{r}^{(w)}z[{k - r - s}]  =
\sum_{r = 0}^{p_1}\sum_{s =
0}^{q_2}\beta_r^{(w)}\alpha_{s}^{(z)}\epsilon^{(z)}[{k - r - s}].
\end{equation*}

Substituting these values into (\ref{eq:recur_n2}), we immediately get
(\ref{eq:recur_lemma}).

The inequality in the above proposition follows from the fact that the
autoregressive and moving-average part may have some common factors.
\end{proof}
\begin{remark}
{\color{black} Note that the process $\{\widetilde{z}[k]\}_{k = 1}^{\infty}$ is
not stationary because of the initial conditions that define the other random
processes $\{w[k]\}_{k = 1}^{\infty}$ and $\{z[k]\}_{k =
1}^{\infty}$. However these edge effects do not affect the asymptotic
distribution of the resulting noise process $\widetilde{z}[k]$.}
\end{remark}
\begin{remark}
{\color{black} While the process $\left\{\widetilde{z}[k]\right\}_{k =
1}^{\infty}$ is not stationary, we note that this non-stationarity has no effect
on the calculation of a lower bound on the
feedback capacity (an asymptotic measure) of the relay channel. Indeed, it is
possible to generate a stationary noise process $\left\{\widetilde{z}[k]\right\}_{k =
\max(p,q)}^{\infty}$ by correctly choosing the initial condition for
the ARMA process. This can be achieved by artificially injecting noise
and/or discarding the first $\max(p,q)$ symbols at the destination and
further imposing the condition that no signal is transmitted in the first
$\max(p,q)$ channel uses, i.e., $x[k] = 0, \forall k \leq \max(p,q)$.}
\end{remark}
\begin{corollary}
For the AWGN processes $\{w[k]\}_{k = 1}^{\infty}$ and $\{z[k]\}_{k =
1}^{\infty}$, the effective noise process $\{\widetilde{z}[k]\}_{k =
1}^{\infty}$ is an ARMA$(L, L - 1)$ Gaussian random process.
\end{corollary}

Having established that the effective noise process is an ARMA$(p,q)$ process
with state space representation as given in (\ref{ss_n}), we next present
a lower bound on the three-terminal relay with destination-source feedback. 

\begin{theorem}\label{arn:lemma2}
If the effective noise $\{\widetilde{z}[k]\}_{k = 1}^{\infty}$ defined in
(\ref{def_n}) is an ARMA$(p,q)$ process having a state space representation as
described in (\ref{ss_n}), a lower bound on the feedback capacity as defined
in (\ref{eq:CfbOpt}) of a three-terminal relay channel is given by
\begin{equation}\label{lb_arma}
R_{\rm LB} = \sup_{\bs, \{h[i]\}_{i = 1}^L}\frac{1}{2}\log\left(1 + (\bs +
\br)^T\bSigma(\bs + \br)\right)
\end{equation}
where $\bs \in \mathbb{R}^{d \times 1}$ such that $\bP - \bq(\bs + \br)^T$ has
no eigenvalue exactly on the unit circle and $\bs^T \bSigma \bs \leq
\rho/\alpha_0^2$ where $\bSigma$ is the maximal solution of the discrete Riccati Algebraic equation
\begin{equation}
\bSigma = \bP \bSigma \bP^T + \bq \bq^T - \frac{\left(\bP\bSigma(\bs + \br) +
\bq\right)\left(\bP\bSigma(\bs + \br) + \bq\right)^T}{1 + (\bs +
\br)^T\bSigma(\bs + \br)},
\end{equation}
the power constraint at the relay in (\ref{eq:lb_opt1}) is satisfied, and the
noise process $\{\widetilde{z}[k]\}_{k = 1}^{\infty}$ is stable~(see Definition
2).
\end{theorem}
\begin{proof}
As shown in the subsection on problem reformulation, we can write the
effective system as
\begin{equation}
\widetilde{\by}  = \bx + \widetilde{\bz},
\end{equation}
such that
\begin{subequations}
\begin{align}
\label{lem1:eq2} \quad E[\bx^T\bx] \leq (N + L)\rho,\\ 
\label{lem1:eq3} \quad \frac{1}{N + L}\textrm{tr}\left(\bH E\left[(\bx +
\bw)(\bx + \bw)^T\right]\bH^T\right) \leq \gamma \rho.
\end{align}
\end{subequations}

 With the fixed set of filter taps, the three-terminal relay problem can be
 viewed as a virtual point-to-point link with the power constraints  given by
 (\ref{lem1:eq2}) and (\ref{lem1:eq3}). We begin by solving for the
 optimal achievable rate for a given set of $\{h[\ell]\}_{\ell = 1}^{L}$.
 
 {\color{black}This is the key result of the work in \cite{Ki10}, and here we
 outline the broad techniques followed in there. The first step in the proof
 begins with characterizing the feedback capacity in its variational form. In
 particular, Theorem 3.2 in \cite{Ki10} states that the feedback capacity for
 a point-to-point link is given by
 \begin{equation*}
 C_{\mathrm{FB}} = \sup_{S_V, B}
 \int_{-\pi}^{\pi}\frac{1}{2}\log\frac{S_{V}(e^{j\omega}) + \left|1 +
 B(e^{j\omega})\right|^2S_{\widetilde{Z}}(e^{j\omega})}{S_{\widetilde{Z}}(e^{j\omega})}\frac{d\omega}{2\pi},
 \end{equation*}
 where $S_{\widetilde{Z}} (e^{j\omega})$ is the power spectral density of the
 noise process $\{\widetilde{z}[k]\}_{k = 1}^{\infty}$ and the maximum is
 performed over all the non-negative power spectral densities
 $S_{V}(e^{j\omega})$ and all strictly causal filters $B(e^{j\omega})$ that
 satisfy the power constraint at the source given by
 \begin{equation*}
 \int_{-\pi}^{\pi}\left(S_{V}(e^{j\omega}) +
 \left|B(e^{j\omega})\right|^2S_{\widetilde{Z}}
 (e^{j\omega})\right)\frac{d\omega}{2\pi} \leq \rho.
 \end{equation*}
 Following the above variational characterization of the feedback capacity, the
 next step involves coming up with the optimal structure for the pair
 $(S_V^{\rm{opt}}, B^{\rm{opt}})$. It is then shown in Theorem 4.6 in
 \cite{Ki10} that without any loss of optimality, one can assume that 
 $S_V^{\rm{opt}} = 0$. This effectively reduces the above expression as
 \begin{equation}\label{eq:var1}
 C_{\mathrm{FB}} = \sup_{B}
 \int_{-\pi}^{\pi}\frac{1}{2}\log\left|1 +
 B(e^{j\omega})\right|^2\frac{d\omega}{2\pi},
 \end{equation}
 such that
 \begin{equation}
 \int_{-\pi}^{\pi}
 \left|B(e^{j\omega})\right|^2S_{\widetilde{Z}}
 (e^{j\omega})\frac{d\omega}{2\pi} \leq \rho.
 \end{equation}
 
As a result, all the effort is now spent on coming up with the optimal causal
feedback filter $B^{\rm{opt}}(e^{j\omega})$. For an ARMA$(p,q)$ noise process
that we consider having the state space representation in (\ref{ss_n}),
the power spectral density is given by $S_{\widetilde{z}}(e^{j\omega}) = \left|
H_{\widetilde{z}}(e^{j\omega})\right|^2,$
where 
\begin{equation*}
H_{\widetilde{z}}(e^{j\omega}) = \alpha_0e^{j\omega}\br^T(\bI -
e^{j\omega}\bP)^{-1}\bq + \alpha_0.
\end{equation*}
The next step in the process is to identify the structure of the optimal coding
scheme at the source. It is shown in \cite{Ki10} that the optimal coding
strategy has to be of the form
 \begin{equation*}
 x[k] = \bs^T\left(\bb[k] - E\left[\bb[k]\big{|}\{\widetilde{y}[i]\}_{i = 1}^{k
 - 1}\right]\right)
 \end{equation*}
 for some $\bs$ such that $\bP - \bq(\bs + \br)^T$ has
no eigenvalue exactly on the unit circle. Intuitively, the above coding
structure ensures that in every new transmission, only new information is
being transmitted which is orthogonal to all the transmissions already done.
Once the optimal structure of the coding scheme has been determined, the calculation
of the optimal rate follows after direct substitution of the values as outlined
in~\cite{Ki10}.
We then maximize this achievable rate over the set of all filter taps while making sure that the constraint in (\ref{lem1:eq3}) is satisfied to obtain the above result. Note that in the above analysis we have assumed that the value of $L$ does not scale
with the change in $N$, i.e., the number of filter taps remain the same even
when the number of transmissions used for the message $\btheta$ increases.} 
\end{proof}

\section{Specialized Results}
In this section, we examine some of the special cases of the generalized relay
network model considered in initial formulation of the problem.

\subsection{Amplify-and-Forward Relay Node}\label{subsec:afs}
In this case the relay network consists of one relay node in total with a
single filter tap $h[1]$. Furthermore, assume that the noise processes
$\{w[k]\}_{k = 1}^\infty$ and $\{z[k]\}_{k = 1}^\infty$ are MA(1) random
processes given by
\begin{subequations}\label{ma_1}
\begin{align}
w[k] &  = \alpha_0^{(w)}\epsilon^{(w)}[k] + \alpha_1^{(w)}\epsilon^{(w)}[k
-1],\\ 
z[k] & = \alpha_0^{(z)}\epsilon^{(z)}[k] + \alpha_1^{(z)}\epsilon^{(z)}[k-1].
\end{align}
\end{subequations}
With the above setting, we have a lower bound on the feedback capacity as given
below.
\begin{lemma}\label{lemma_MA}
A lower bound on the feedback capacity of a three-terminal relay channel with
one filter tap~($C_{\rm{FB}, 1}$ in (\ref{eq:CfbOpt})) with source-to-relay and
source-to-destination noise evolving as MA(1) noise process~(see (\ref{ma_1})) is given by
\begin{equation}
R_{\rm LB} = \sup_{h[1]}\left(-\log \xi_0\right),
\end{equation}
where $\xi_0$ is the unique positive root of the quartic polynomial
\begin{equation}
\frac{\rho}{\alpha_0^2}\xi^2 = \frac{(1 - \xi^2)(1 + \psi
\alpha_1/\alpha_0\xi)^2}{(1 + \psi h[1]\xi)^2},
\end{equation}
with $\psi = sgn(h[1] - \alpha_1/\alpha_0)$ and $h^2[1] \leq \min(\gamma
\frac{P}{P + \sigma_w^2},1)$.
\end{lemma}
\begin{proof}
For MA(1) noise processes in (\ref{ma_1}) and only one filter tap at the
relay, the effective noise process $\{\widetilde{z}[k]\}_{k = 1}^{\infty}$ can
be described by an ARMA(1,1) noise process
\begin{equation}
\label{eq:arma1_eff}
\widetilde{z}[k] + h[1] \widetilde{z}[k - 1] = \alpha_0 \epsilon[k] + \alpha_1 \epsilon[{k - 1}],
\end{equation}
where
\begin{align*}
\alpha_0^2 + \alpha_1^2 & = 1 + h^2[1]\sigma_w^2,\\
\alpha_0\alpha_1 & = \alpha_0^{(z)}\alpha_1^{(z)} +
h^2[1]\alpha_0^{(w)}\alpha_1^{(w)}.
\end{align*}

Furthermore, the power constraint at the relay in (\ref{eq:lb_opt1}) can be
upper bounded by
\[h^2[1] \leq \gamma \frac{\rho}{\rho + \sigma_w^2}.\]
Note that we also require that $h^2[1] < 1$ to ensure that the noise process
$\{\widetilde{z}[k]\}_{k = 1}^{\infty}$ is stable. Therefore, we have the
overall constraint on the filter tap as $h^2[1] \leq \min(\gamma\frac{\rho}{\rho
+ \sigma_w^2},1)$. 

{\color{black}The spectral density of the overall noise process in
(\ref{eq:arma1_eff}) is given by
\begin{equation}
S_{\widetilde{z}}(e^{j\omega}) = \left|\frac{\alpha_0 + \alpha_1e^{j\omega}}{1 +
h[1]e^{j\omega}}\right|^2.
\end{equation}
As outlined before in the proof of Theorem \ref{arn:lemma2}, the main goal is to
come up with a strictly causal optimal feedback filter
$B(e^{j\omega})$ that achieves the capacity. It has been shown in \cite{Ki10}
that an optimal feedback filter is of the form
\begin{equation}
B(e^{j\omega}) = \frac{1 + h[1]e^{j\omega}}{\alpha_0 + \alpha_1
e^{j\omega}}\cdot \frac{\eta e^{j\omega}}{1 - \psi \xi e^{j\omega}},
\end{equation}
where $\xi \in (0,1)$ is chosen so that
\begin{equation}
\eta = \frac{\xi^2 - 1}{\psi \xi}\cdot\frac{\alpha_0 + \alpha_1\psi \xi}{1 +
h[1]\psi \xi} = -\rho \psi \xi\frac{1 + h[1] \psi \xi}{\alpha_0 + \alpha_1
\psi \xi}.
\end{equation}
Plugging the value of this optimal filter into the variational form of capacity
as outlined in (\ref{eq:var1}) immediately gives us the above result.
}

\end{proof}

It is worth pointing out that the rate $R_{\rm LB}$ in the above lemma is
achievable by a variation of the celebrated SK scheme~\cite{ScKa66, Sc66} as
outlined in \cite{Ki10}. Note that the above lemma contains as a special case
the lower bound in \cite{BrWi09} for the relay channel with white noise (i.e.,
$\alpha_1^{(w)} = \alpha_1^{(z)} = 0$ in (\ref{ma_1})) which was shown to
outperform the more sophisticated block-Markov strategies~\cite{CoGa79} for a
wide selection of available power~($\gamma$) at the relay. For AWGN, the
effective noise process in our formulation reduces to an AR(1) process for
which the scheme proposed in \cite{Bu69} achieves the lower bound.

\begin{corollary}
(Theorem 5 in \cite{BrWi09}): A lower bound on the three-terminal relay
channel with AWGN processes is given by $R_{\rm LB} = \sup_{h[1]}(-\log \xi_0)$,
where $\xi_0$ is the unique positive root of the quartic polynomial
\begin{equation*}
\frac{P}{1 + \sigma_w^2 h^2[1]}\xi^2 = \frac{(1 - \xi^2)}{(1 + |h[1]|\xi)^2},
\end{equation*}
with $h^2[1] \leq \min(\gamma
\frac{P}{P + \sigma_w^2},1).$
\end{corollary}

\begin{figure}
\centering
\includegraphics[scale = 0.45]{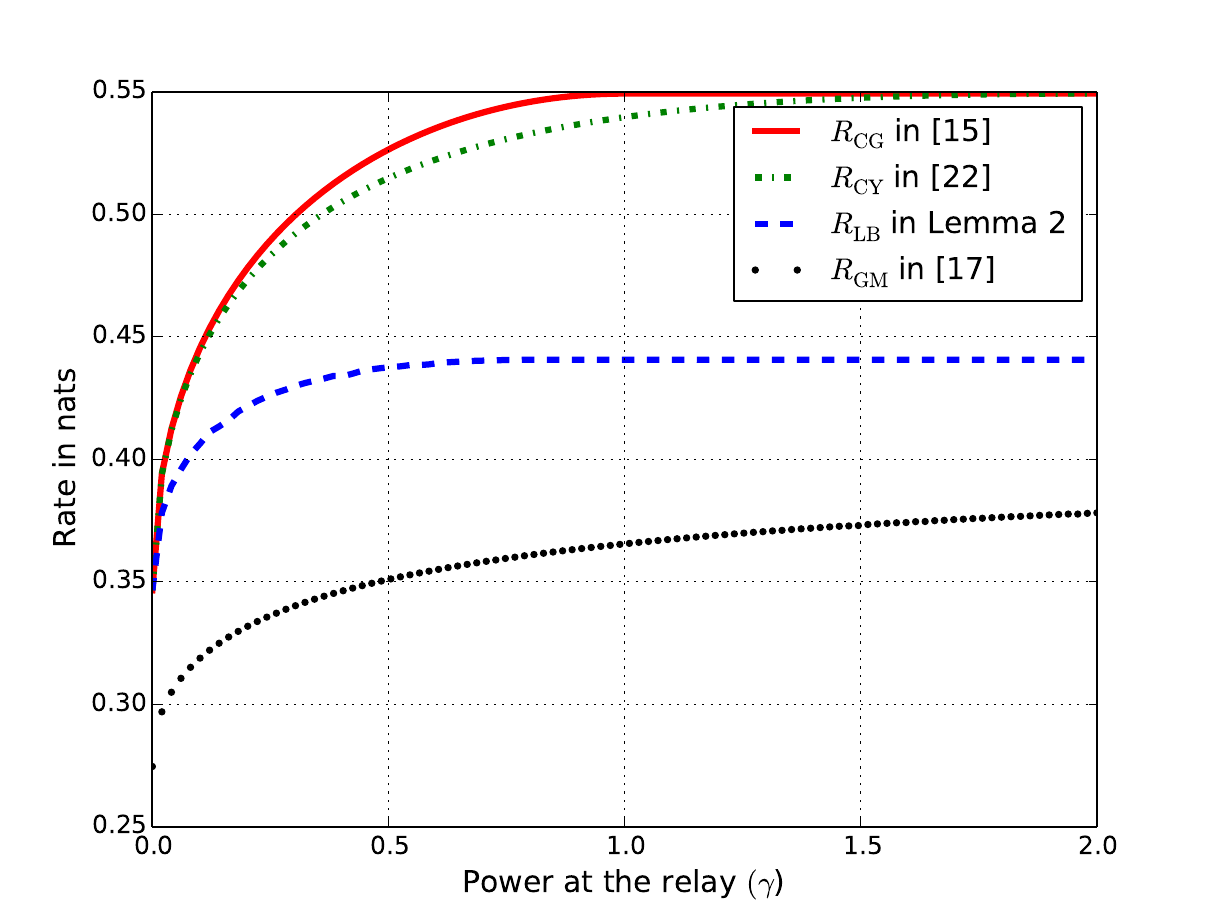}
\caption{Comparison of achievable rates for various
schemes under AWGN with power available at the relay. The parameters
used for simulation were $\rho = 1$ and $\sigma_w^2 = 1$.}
\label{sim:cr}
\end{figure}
 
\begin{figure}
\centering
\includegraphics[scale = 0.45]{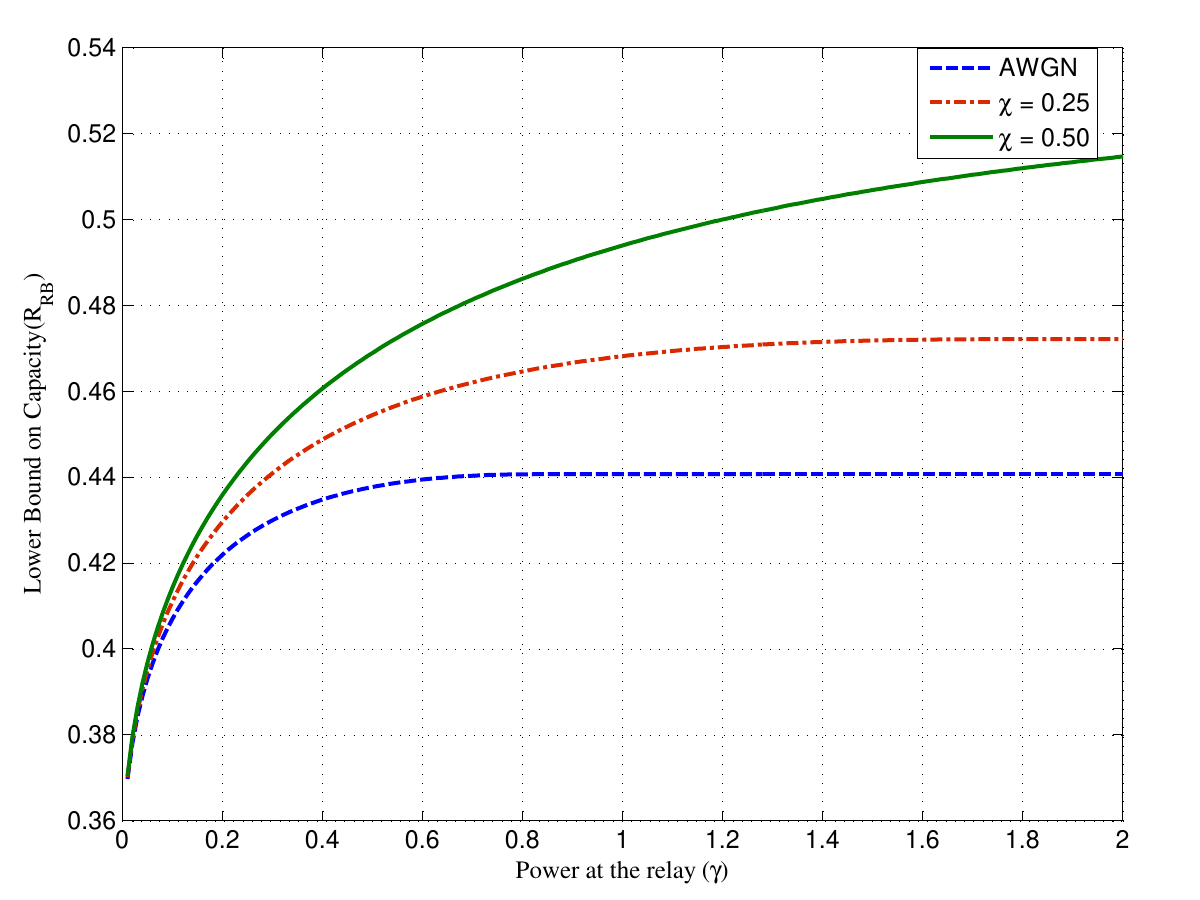}
\caption{Plot of variation of lower bound on capacity~($R_{\rm LB}$) with power
available at the relay. The parameters used for simulation were $\rho = 1$ and
$\sigma_w^2 = 1$. The blue~(dashed), red~(dashed-dot) and green~(solid)
curves correspond to $\chi$ of 0, 0.25, and 0.50 in (\ref{sim1}), respectively.}
\label{sim:ma}
\end{figure} 
{\color{black} For illustration, we compare the lower bound in Lemma
\ref{lemma_MA} ($R_\mathrm{LB}$) with various bounds and schemes available in
the literature. In Figure \ref{sim:cr}, we compare the rates for AWGN relay
channel~(both $\{w[k]\}_{k = 1}^\infty$ and $\{z[k]\}_{k = 1}^\infty$ are
assumed to be white) with $\rho = 1$ and $\sigma_w^2 = 1$ for varying power at
the relay node~($\gamma$).
$R_\mathrm{CG}$ denotes the max-flow min-cut upper bound derived in \cite{CoGa79}, $R_\mathrm{CY}$ the
capacity of a three terminal amplify-and-forward relay channel with 
source information available with {zero-delay~(non-causally)} at the
relay node in \cite{ChYo12}, while $R_\mathrm{GM}$ plots the achievable rate for a linear scheme at the relay node with blocklength of
two as proposed in \cite{GaMo06}. It can be seen from Figure \ref{sim:cr} that
our proposed scheme can offer significant improvements over the open loop linear
coding scheme due to the availability of feedback link from the destination to
the source. For $\gamma = 0.5$, we obtain about $25.7\%$ improvement in the
achievable rate over the scheme in \cite{GaMo06}. However, Figure \ref{sim:cr}
also shows that the achievable rate with non-causal
information available at the relay is sufficiently higher than our proposed
lower bound with feedback.}

We next present numerical results to demonstrate the rate improvements achieved
due to the noise {with memory} in the source-to-relay link. We
assume that the source-to-relay noise process $\{w[k]\}_{k = 1}^\infty$ is an MA(1) noise process as
 \begin{equation}\label{sim1}
w[k] = \sigma_w\chi \epsilon^{(w)}[k] + \sigma_w \sqrt{1 -
\chi^2}\epsilon^{(w)}[k - 1],
\end{equation}
with $0 \leq \chi \leq 1$. {The source-to-destination link
$\{z[k]\}_{k = 1}^\infty$ is assumed to be AWGN with $\cN(0,1)$.}

Figure \ref{sim:ma} plots the variation of the lower bound on the three-terminal
relay capacity as a function of the relay transmit power~($\gamma$). A higher
value of $\gamma$ implies more power available at the relay node. A
value of $\gamma = 1$ means that both the source and the relay have the
same amount of available average power; whereas $\gamma = 0$ completely shuts-off the relay with no
power available. The simulations were performed with $P = 1$, $\sigma_w^2 =
1$, and the value of $\gamma$ in the interval $[0,2]$.
The AWGN curve corresponds to the case when $\chi = 0$ for source-to-relay
noise~\cite{BrWi09}. As pointed out before, $\gamma = 0$ implies the absence of
a relay altogether; this explains the convergence of all the three curves to
the same lower bound~(in fact capacity) as $\gamma \rightarrow 0$. It is seen
that the availability of more power at the relay with AWGN noise between the
source-to-relay pair can lead to an improvement of $19\%$ over the
point-to-point link with the given simulation parameters.

As the {memory} factor $\chi$ is increased in the source-to-relay
pair, we see even further improvement in the achievable rate. For $\chi =
0.25$, the rate can increase by $28\%$ to a maximum rate of 0.472
nats/channel use. The achievable rate shows a $43\%$ improvement
over the point-to-point {AWGN} link capacity for $\chi = 0.50$. This
suggests that significant gains in the rates can be achieved when noise
{ with memory} is present at the source-to-relay pair and additional
power~($\gamma > 0$) is available at the relay node.

\subsection{AWGN Relay with Two Taps}
Now we look at the three-terminal relay node with two filter taps $\{h[1],
h[2]\}$.  The noise vectors $\bw$ and $\bz$ are assumed to be white, i.e.,
$E[\bz \bz^T] = \bI$ and $E[\bw \bw^T] = \sigma_w^2 \bI$. To satisfy the
stability constraints on the filter taps, we require that the polynomial
$\phi(b) = b^2 + h_1b + h_2$, has both the roots inside the unit circle. This
corresponds to the following conditions on choosing $h[1]$ and $h[2]$:
\begin{equation}\label{eq:2tap1}
1 - |h[1]| + h[2]  > 0, \quad |h[2]| < 1.
\end{equation}
Furthermore, the power constraint at the relay puts an additional upper bound on
the taps 
\be \label{eq:2tap2}
h^2[1] + h^2[2] \leq \frac{\gamma \rho}{\sigma_w^2}.
\ee
{Note that the equality in the above equation will hold only in
the unlikely scenario of the source not transmitting any message at all. As a
result, (\ref{eq:2tap2}) only serves the purpose of limiting the
range of values that the taps $h[1]$ and $h[2]$ can have without making
any relaxation to the original problem}.
As noted before, numerical optimization of the lower bound on the feedback relay capacity in Lemma \ref{arn:lemma2} is not obvious for the case where more than one filter tap is available at the relay. Therefore, we try to
compute the gains of two taps using the $(N,L)~$block feedback capacity as
described in Proposition \ref{lemma_convex}.

We begin by arbitrarily setting the coefficients of the filter taps $\{h[1],
h[2]\}$ to satisfy the stability and power constraints at the relay given by
(\ref{eq:2tap1}) and (\ref{eq:2tap2}). Now having chosen the coefficients, the
optimization in (\ref{eq:lb_opt1}) reduces to a convex optimization problem as
outlined by Proposition \ref{lemma_convex} which can be solved numerically. For
each value of $\gamma$, we optimize by arbitrarily generating 1000 candidate
filter taps and then taking the maximum achievable rate over all the generated
filter possibilities. The result of such an analysis is shown in Table \ref{table1}.

Table \ref{table1} shows the lower bound on the $(N, L)-$block feedback capacity
for $N = 20, \rho = 1,$ and $\sigma_w^2 = 0.1$. In the absence of a relay, the
maximum achievable rate is simply the capacity of the point-to-point
{AWGN} link.
With the help of just one-filter tap, an improvement of up to $65.6\%$ is achievable. This
happens when the power factor at the relay is given by $\gamma = 1.1$.

The availability of additional power at the relay~($\gamma > 1.1)$, however,
does not lead to any further improvement in the achievable rate with a simple
amplify-and-forward scheme. The presence of two taps allows us to exploit
additional power available at the relay for rate improvement. In Table
\ref{table1}, we see that even with an increase in $\gamma$ from $1.1$ to $1.3$
leads to an improvement in the lower rate by about $5\%$. Furthermore with more
power available at the relay node, the lower bound on the three-terminal relay
shows an improvement of almost 100\% over no power at the relay. This example
suggests that an increase in the number of taps at the relay can lead to better
achievable rates if substantial power is available at the relay.
\begin{table} 
\caption{Lower Bound on feedback capacity with two filter taps for $N = 20, \rho
= 1,$ and $\sigma_w^2 = 0.1$}
\begin{center}
\begin{tabular}{|c|c|c|c|c|}
\hline
\hline
$\gamma$ & $h_1$ & $h_2$ & $R_{\rm{LB,N}}$ & \% change\\
\hline
$0$ & $0.00$ & $0.00$ & $0.346$ & --\\
$1.1$ & 1.00 & 0.00 & 0.573 & $65.6\%$ \\
1.3 & 1.04 & 0.14 & 0.589 & $70.2\%$\\
1.8 & -1.21 & 0.26 & 0.610 & $76.3\%$\\
2.5 & 1.36 & 0.39 & 0.633 & $82.9\%$\\
5.0 & -1.67 & 0.88 & 0.667 & $92.8\%$\\
\hline
\hline
\end{tabular}
\end{center}
\label{table1}
\end{table} 

\begin{figure}
\centering
\includegraphics[scale = 0.25]{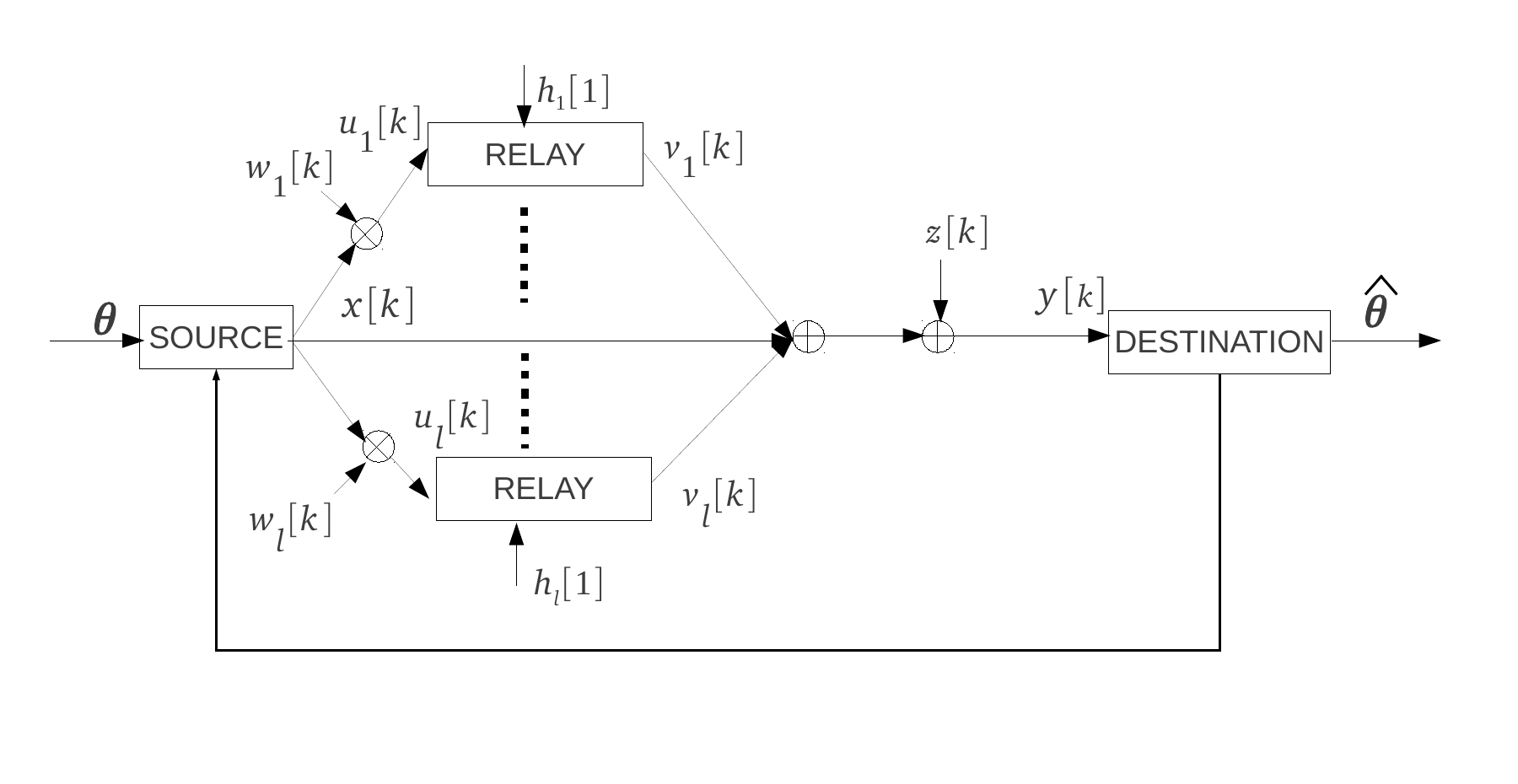}
\caption{System model for the extended relay model with $|\cV|$
amplify-and-forward relays in parallel each having gain $\alpha_i$.}
\label{ONF:parallel}
\end{figure} 

\subsection{Amplify-and-Forward Relays in Parallel}
Consider the parallel network as shown in Figure \ref{ONF:parallel}. In this
case we assume that the noise process $\{w_i[k]\}$ for the $i^{th}$ relay node
at time instant $k$ is white with $\cN(0, \sigma_i^2)$. It is assumed that each
of the noise processes is independent of the other. Furthermore, each relay node has a single
tap given by $h_i[1]$ with the power constraint
\be \nonumber \sum_{k = 1}^N E[v_i^2[k]] \leq
\gamma_i N\rho, \quad 1 \leq i \leq |\cV|.
\ee

Under this network configuration, a lower bound on the feedback capacity is
presented in the following lemma.

\begin{lemma}
{A lower bound on the maximum achievable rate
of the amplify-and-forward relay network with additive white Gaussian noise as
depicted in Figure \ref{ONF:parallel}} is given by $R_{\rm LB} = \sup_{\{h_i[1]\}_{i = 1}^{|\cV|}}(-\log \xi_0)$, where $\xi_0$ is the unique positive root of the quartic polynomial
\begin{equation*}
\frac{\rho}{1 + \sum_{i = 1}^{|\cV|} h_i^2[1]\sigma_i^2}\xi^2 = \frac{(1 -
\xi^2)}{(1 + |\sum_{i = 1}^{|\cV|} h_i[1]|\xi)^2},
\end{equation*}

with $h_i^2[1] \leq \gamma_i \frac{\rho}{\rho + \sigma_i^2}, 1 \leq i \leq
|\cV|, \left(\sum_{i = 1}^{|\cV|} h_i[1]\right)^2 \leq 1.$
\end{lemma}

\begin{proof}
Under this parallel scheme of transmission,
the received signal at the destination is given by
\be
\nonumber
y[k] = x[k] + \left(\sum_{i = 1}^{|\cV|} h_i[1]\right)x[k - 1] + \sum_{i =
1}^{|\cV|} h_i[1] w_i[k - 1] + z[k].
\ee
In vector form, it can again be written as
\be \nonumber
\by = (\bI + \bH) \bx + \breve{\bz},
\ee
with 
\be \nonumber
\bH = \left[ \begin{array}{ccccccc}
0 & 0 & 0 & \ldots & \ldots & \ldots & \ldots\\
\sum_{i = 1}^{|\cV|} h_i[1] & 0 & 0 & \ldots & \ldots & \ldots & \ldots\\
\vdots & \vdots & \ddots & \ddots & \vdots & \vdots & \vdots \\
\vdots & \vdots & \ddots & \ddots & \vdots & \vdots & \vdots \\
0 & \ldots & 0 & 0 & \ldots & \sum_{i = 1}^{|\cV|} h_i[1] & 0\\
\end{array}\right],
\ee
and $\breve{\bz}$ is a Gaussian vector with zero mean and
$E[\breve{\bz}\breve{\bz}^T] = \left(1 + \sum_{i = 1}^{|\cV|}
h_i^2[1] \sigma_i^2\right)\bI$. After inversion, the effective noise vector
becomes $\widetilde{\bz} = (\bI + \bH)^{-1}\breve{\bz}$. Hence, the
effective noise process $\{\widetilde{z}[k]\}_{k = 1}^\infty$ satisfies 
\be
\nonumber \widetilde{z}[k] + \left(\sum_{i = 1}^{|\cV|} h_i[1]\right) \widetilde{z}[k - 1]
= \sqrt{\left(1 + \sum_{i = 1}^{|\cV|} h_i^2[1] \sigma_i^2\right)}\epsilon[k],
\ee
which is again an AR(1) process. Therefore a lower bound on the capacity of
the system in Figure \ref{ONF:parallel} is given by the above lemma.
\end{proof}
\begin{figure}
\centering
\includegraphics[scale = 0.45]{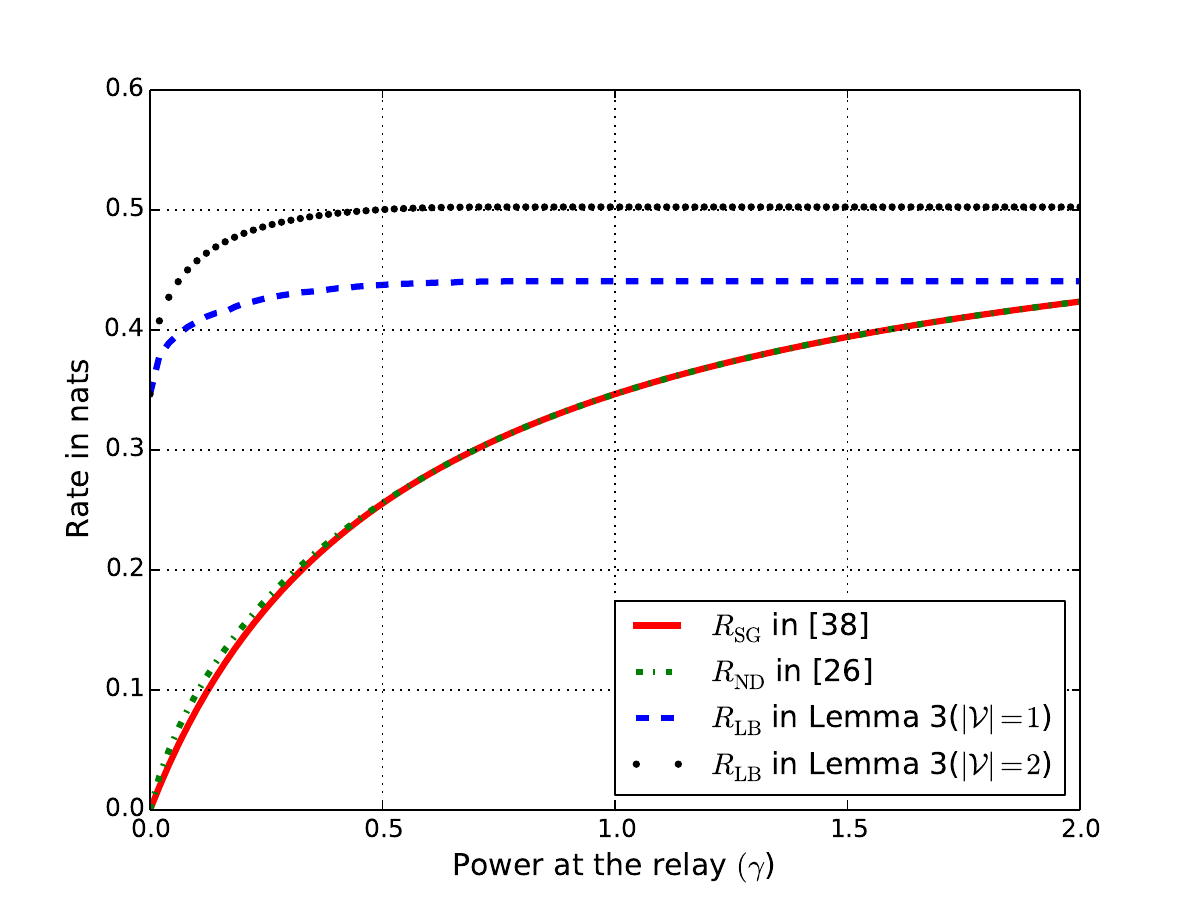}
\caption{Comparison of achievable rates of various schemes for symmetric AWGN
parallel relay channel. The parameters used for simulation were $\rho = 1,
\sigma_i^2 = 1,$ and $\gamma_i = \gamma, 1 \leq i \leq |\cV|$.}
\label{sim:parallel}
\end{figure}
{\color{black}In Figure \ref{sim:parallel}, we compare the rates for
symmetric AWGN parallel relay channel with $\rho = 1, \sigma_i^2 = 1$, and
$\gamma_i = \gamma$ for $1 \leq i \leq |\cV|$. The rates of the two schemes that
we compare against, $R_\mathrm{SG}$ and $R_\mathrm{ND}$, are both open
loop schemes with two relay nodes in parallel, and no direct link available
between the source and the destination. $R_\mathrm{SG}$ in \cite{Sc01} denotes
the maximum achievable rate using amplify-and-forward scheme, while
$R_\mathrm{ND}$ the lower bound on the capacity of AWGN parallel relay channel
with 2 nodes using linear relaying in \cite{NiDi13,XuKi14}. The expression for $R_\mathrm{ND}$ as
described in \cite{XuKi14} is given by
\begin{equation*}
R_\mathrm{ND} = \max_{0 \leq g \leq
1}gC\left(\frac{4\gamma(\frac{\rho}{g})^2}{1 +
\frac{\rho}{g} + 2\gamma\frac{\rho}{g}}\right)
\end{equation*}
where $C(x) = 1/2\log(1 + x).$ It can be seen from Figure \ref{sim:parallel}
that our proposed lower bound in Lemma 3 with both one relay node
(original three terminal relay channel) and two relay nodes perform much
better than the open-loop schemes available in the literature. The gains are substantial
especially when there is only a low power available at the relay
nodes.}

\begin{figure}
\centering
\includegraphics[scale = 0.20]{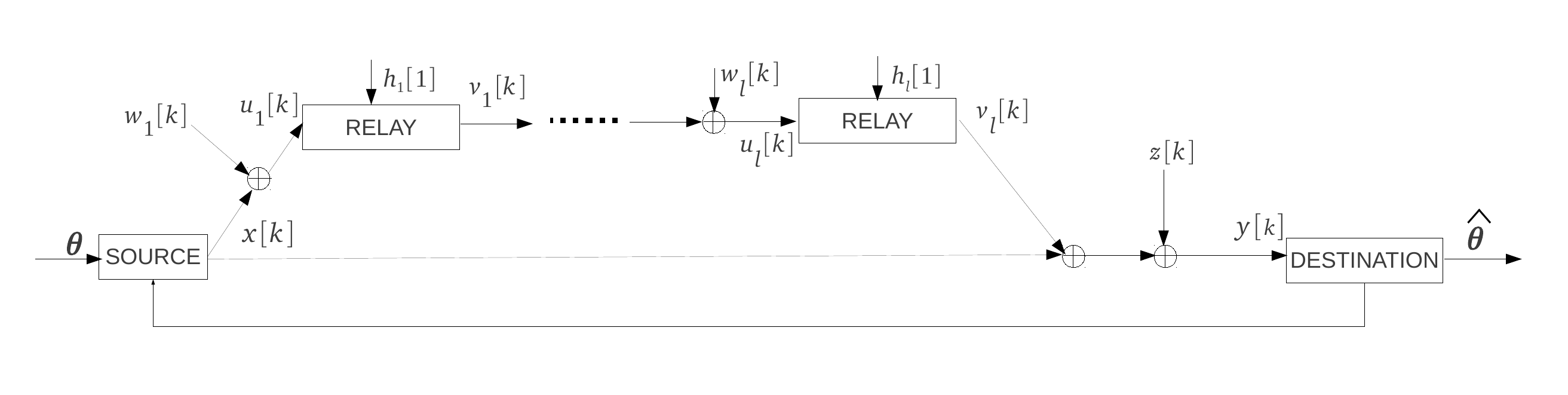}
\caption{System model for extended relay model with $\ell$ amplify-and-forward
relays in series each having gain $h_i[1]$.}
\label{ONF:series}
\end{figure} 
\subsection{Amplify-and-Forward Relays in Series}

In the series configuration, we consider the network as shown in Figure
\ref{ONF:series} with similar assumptions about the {noise
processes~(AWGN)} and power at the relays as used before. Under this setting,
the received signal at the destination is given by
\begin{align*}
y[k] & = x[k] + \left(\prod_{i = 1}^{|\cV|}h_i[1]
\right)x[k - |\cV|]\\
&  \qquad + \sum_{j = 1}^{|\cV|} \left(\prod_{i =
j}^{|\cV|}h_i[1] \right)w_j[k - (\ell + 1 - j)] + z[k].
\end{align*}

Hence the noise process, $\{\widetilde{z}[k]\}_{k = 1}^{\infty}$ can be
described as \be \nonumber
\widetilde{z}[k] + \left(\prod_{i = 1}^{|\cV|}h_i[1] \right)\widetilde{z}[k
- \ell] = \sqrt{\left(1 + \sum_{j = 1}^{|\cV|} \left(\prod_{i = j}^{|\cV|}
h_i^2[1]\right) \sigma_j^2\right)}\epsilon[k],
\ee
where the effective noise process in now just an AR($|\cV|$) process.
Furthermore, the power constraints at the relays can be upper bounded as
\be \label{series:eq1}
h_i^2[1] \leq \frac{\gamma_{i}\rho}{\gamma_{i - 1}\rho + \sigma_i^2}, \quad i
= 1, \ldots, |\cV|,
\ee
where $\gamma_0 = 1$. {To maintain the stability of the
effective noise process} we require that 
\be \label{series:eq2}
\prod_{i = 1}^{|\cV|}h_i^2[1] < 1.
\ee
Now, a lower bound on the capacity can be achieved by invoking Lemma
\ref{arn:lemma2} for the AR($|\cV|$) process with the filter tap coefficients
limited to the values as described in (\ref{series:eq1}) and (\ref{series:eq2}).

\section{Noisy Channel Output Feedback}\label{sec:NCOF_N_2}

\subsection{General Framework}
So far our analysis has been concentrated on the scenario where the feedback
link is noiseless, i.e., $\sigma_n^2 = 0$. In this section, we look at the scenario when
the channel output feedback is in fact noisy. {All the noises in
this section will be assumed to be additive white Gaussian in nature.} This
leads to the situation where the source no longer has exact knowledge about the state of decoding at the destination. We will develop a linear coding framework assuming that the source
and the destination can perform only linear operations. This is motivated by the
fact that for the noiseless case, the lower bound reported in the previous
section is achievable using linear encoding and decoding.

In the presence of noise, the metric that we are interested in optimizing
is the post-processed signal-to-noise ratio~(SNR) over the length of
transmission of a single symbol. Furthermore to obtain a non-zero achievable
rate, our proposed coding scheme can be used in concatenated fashion as outlined
in \cite{ZaDa11} for point-to-point communication. It can also be
incorporated in the network protocols using variants of automatic repeat
request~(ARQ).

\subsection{System Model and Example with $N = 2$}

We study a scheme for the case of two channel uses, i.e., $N = 2$. Hence the
relay node has only one filter tap, i.e., $L = 1$. In Figure \ref{fig:sm}, at
time instance $k = 1$, the signals at various nodes are given by
\begin{subequations}
\begin{align}
x[1] & = g_1 \theta\\
u[1] & = x[1] + w[1] =  g_1 \theta + w[1]\\
v[1] & = 0\quad \textrm{(no signal transmitted)}\\
y[1] & = x[1] + z[1] = g_1 \theta + z[1].
\end{align}
\end{subequations}

Due to the availability of a feedback link, the source has access to the
additional side information $y[1] + n[1]$ before transmitting $x[2]$. However,
this additional side-information is equivalent to having knowledge of $z[1] +
n[1]$. Therefore, at time instance $k = 2$, the signals transmitted at various nodes are
\begin{subequations}
\begin{align}
\!\!\!\!\!\!x[2] & = g_2 \theta + f_{21}\left(z[1] + n[1]\right)\\
\!\!\!\!\!\!u[2] & = 0 \quad \textrm{(no signal received)}\\
\!\!\!\!\!\!v[2] & = h_1u[1] = h_1(x[1] + w[1]) = h_1(g_1 \theta + w[1])\\
\nonumber \!\!\!\!\!\!\!\!y[2] & = x[2] + v[2] + z[2]\\
\!\!\!\!\!\!& = (g_2 + h_1g_1)\theta + f_{21}(z[1] +
n[1]) + h_1w[1] + z[2].
\end{align}
\end{subequations}
In vector form, the two transmissions can be combined into the form in
(\ref{eq:2T}).
\begin{figure*}
\begin{align}\label{eq:2T}
\nonumber
\left[\begin{array}{c}
y[1] \\
y[2] \end{array}\right] & = \left[\begin{array}{c}
g_1\\g_2 + h_1g_1\end{array}\right]\theta + \left[\begin{array}{cc}
1 & 0\\f_{21} & 1\end{array}\right]\left[\begin{array}{c}
z[1] \\ z[2] \end{array}\right] + \left[\begin{array}{cc}
0 & 0\\f_{21} & 0\end{array}\right]\left[\begin{array}{c}
n[1] \\ n[2] \end{array}\right] + \left[\begin{array}{cc}
0 & 0\\h_1 & 0\end{array}\right]\left[\begin{array}{c}
w[1] \\ w[2] \end{array}\right]\\
\by &= \bg \theta + (\bI + \bF)\bz + \bF \bn + \bB \bw.
\end{align}
\end{figure*}

Using a linear estimator at the destination, the estimate of $\theta$ is
given by $\widehat{\theta} = \br^T \by.$
Now the post-processed SNR at the destination is
\begin{equation}
\textrm{SNR} = \frac{|\br^T\bg|^2}{\br^T\bC \br},
\end{equation}
where $\bC = (\bI + \bF)(\bI + \bF)^T + \sigma_n^2\bF\bF^T + \sigma_w^2 \bB
\bB^T$. Using the Cauchy-Schwartz inequality, it is clear that the
post-processed SNR is maximized by choosing the optimal linear
estimator~\cite{Ka93}.
This reduces the post-processed SNR expression to $\bg^T
\bC^{-1}\bg.$
Furthermore in this development, we assume a per transmission power
constraint at the source is $E[x^2[k]] \leq \rho, \quad k = 1,2.$

We can now describe the overall optimization problem as
\begin{subequations}\label{opt1}
\begin{align}
\max _{g_1, g_2, f_{21}, h_1}& \quad \bg^T \bC^{-1} \bg \\
\nonumber \textrm{such that} \\
\label{c1} &g_1^2 \leq \rho\\
\label{c2}& g_2^2 + (1 + \sigma_n^2)f_{21}^2 \leq \rho\\
\label{c3}& h_1^2 \leq \frac{\gamma \rho}{\rho + \sigma_w^2}.
\end{align}
\end{subequations}
\noindent Substituting the values of $\bg$ and $\bC$ in the post-processed SNR
expression, we obtain

\begin{align}
\nonumber \bg^T \bC^{-1} \bg & = \frac{g_1^2\left(1 + (1 + \sigma_n^2)f_{21}^2 +
\sigma_w^2h_1^2\right) - 2g_1(h_1g_1 + g_2)f_{21}}{1 +
\sigma_n^2f_{21}^2+ \sigma_w^2h_1^2}\\
\nonumber & \qquad + \frac{(h_1g_1 + g_2)^2}{1 +
\sigma_n^2f_{21}^2+ \sigma_w^2h_1^2}\\
\label{SNR2_op1}& = g_1^2 + \frac{\left(g_1(h_1 - f_{21}) + g_2\right)^2}{{1 +
\sigma_n^2f_{21}^2+ \sigma_w^2h_1^2}}.
\end{align}
We perform the above optimization in two steps:
\begin{itemize}
  \item Optimization over $g_1$.
  \item Followed by joint optimization over $g_2, f_{21},$ and $h_1$.
\end{itemize}

\subsection{Optimization of Post-Processed SNR}

\subsubsection{Optimization over $g_1$}\label{opt_g1}
This optimization is trivial. It is obtained at the boundary point with $g_1 =
\sqrt{\rho}$ in (\ref{SNR2_op1}). Hence the $\textrm{SNR}$ is given by
\begin{equation}\label{SNR_g1_opt}
\textrm{SNR}  = \bg^T \bC^{-1} \bg = \rho + \frac{\left(\sqrt{\rho}(h_1 - f_{21}) +
g_2\right)^2}{{1 + \sigma_n^2f_{21}^2+ \sigma_w^2h_1^2}}.
\end{equation}

\subsubsection{Optimization over $g_2$, $f_{21}$ and $h_1$}\label{opt_rest}
To perform joint optimization over $g_2$, $f_{21}$ and $h_1$, we begin by
writing down the Karush-Kuhn-Tucker~(KKT) conditions~\cite{BoVa06} for the
optimal solution in (\ref{opt1}):
\begin{figure*}
\begin{subequations}\label{eq:KKT}
\begin{align}
\label{KKT1} -\frac{\sqrt{\rho}(h_1 - f_{21}) + g_2}{1 + \sigma_n^2f_{21}^2 +
\sigma_w^2 h_1^2} + \mu_2 g_2  = 0,\\
\label{KKT2} \frac{\left(1 + \sigma_n^2f_{21}^2 + \sigma_w^2
h_1^2\right)\sqrt{\rho}\left(\sqrt{\rho}(h_1 - f_{21}) + g_2\right) +
\left(\sqrt{\rho}(h_1 - f_{21}) + g_2\right)^2\sigma_n^2f_{21}}{\left(1 +
\sigma_n^2f_{21}^2 + \sigma_w^2 h_1^2\right)^2} + \mu_2(1 + \sigma_n^2)f_{21}  =
0,\\
 \label{KKT3} -\frac{\left(1 + \sigma_n^2f_{21}^2 + \sigma_w^2
h_1^2\right)\sqrt{\rho}\left(\sqrt{\rho}(h_1 - f_{21}) + g_2\right) -
\left(\sqrt{\rho}(h_1 - f_{21}) + g_2\right)^2\sigma_w^2h_1}{\left(1 +
\sigma_n^2f_{21}^2 + \sigma_w^2 h_1^2\right)^2} + \mu_3  = 0.
\end{align}
\end{subequations}
\end{figure*}
Note that $\mu_2$ and $\mu_3$ are non-negative KKT multipliers associated with
the constraints (\ref{c2}) and (\ref{c3}), respectively. Based on these KKT
conditions, we present the following lemma.

{\begin{lemma}\label{lemma_NR}
The post-processed SNR in (\ref{opt1}) is maximized when $g_2 > 0,
f_{21} < 0,$ and $h_1 > 0$.
\end{lemma}
\begin{proof}
Looking at the expression in (\ref{SNR_g1_opt}), we know that post processed
SNR is maximized  with $sign(h_1 - f_{21}) = sign(g_2)$. WLOG we can assume that
an optimal post-processed SNR will have $g_2$ and $h_1$ as non-negative values
while $f_{21}$ will have non-positive value. Also, note that in the event of any
one of them being zero, from KKT conditions in (\ref{eq:KKT}), we immediately
get $\left(\sqrt{\rho}(h_1 - f_{21}) + g_2\right) = 0$. However, this is a
minimizer of the SNR in (\ref{SNR2_op1}). Hence we conclude that the maximum post-processed SNR will
  have $g_2 > 0, f_{21} < 0$, and $h_1 > 0$.
\end{proof}
}
Using (\ref{KKT1}) and (\ref{KKT2}), we can express $g_2$ in terms of
$f_{21}$ and $h_1$ as 
\begin{equation}\label{eq:g2}
g_2 = -\frac{1}{\sqrt{\rho}}\left(\frac{\sigma_n^2\rho + (1 + \sigma_n^2)(1 +
\sigma_w^2 h_1^2)}{1 + \sigma_w^2h_1^2 + \sigma_n^2 f_{21}h_1}\right)f_{21}.
\end{equation}

\noindent Furthermore we observe that the second constraint (\ref{c2}) in the
optimization problem should always be satisfied with equality. If this was not true, we can
allocate the additional power to $g_2$ to boost the overall post-processed SNR
at the destination in (\ref{SNR_g1_opt}). Hence from (\ref{c2}), we have
\begin{equation}\label{eq:f_21}
\frac{1}{\rho}\left(\frac{\sigma_n^2\rho + (1 + \sigma_n^2)(1 +
\sigma_w^2 h_1^2)}{1 + \sigma_w^2h_1^2 + \sigma_n^2 f_{21}h_1}\right)^2f_{21}^2
+ (1 + \sigma_n^2)f_{21}^2 = \rho.
\end{equation}
Note that the left hand side of the above equation monotonically increases as we
decrease $f_{21}$, thereby implying a unique value of $f_{21}$ for a given
$h_1$. From the expression for $g_2$, we know that the region of interest is
limited to  $\{f_{21} : 1 + \sigma_w^2h_1^2 + \sigma_n^2 f_{21}h_1 > 0\}$.

Now, there are two possibilities, one in which all the power is used at the
relay, and the other in which only a fraction of it is used. {The
scenarios are:}
\begin{itemize}
  \item $h_1 < \sqrt{\frac{\gamma \rho}{\rho + \sigma_w^2}}:$ If the full power is not
  used at the relay, this would imply that $\mu_3 = 0$. Solving the condition in (\ref{KKT3}), we
immediately get
\begin{equation}
h_1 = \frac{\sqrt{\rho}\left(1 +
\sigma_n^2f_{21}^2 + \sigma_w^2 h_1^2\right)}{\sigma_w^2\left(\sqrt{\rho}(h_1 -
f_{21}) + g_2\right)} = \frac{\sqrt{\rho}}{\sigma_w^2}\frac{1}{\mu_2 g_2}.
\end{equation}
Equivalently, 
\begin{equation}\label{eq:h_1}
h_1 = \frac{\sqrt{\rho}\left(1 +
\sigma_n^2f_{21}^2\right)}{\sigma_w^2\left(g_2 - \sqrt{\rho}f_{21}\right)}.
\end{equation}
{Now using (\ref{eq:f_21}) and (\ref{eq:h_1}) in an iterative
manner, we can solve for an optimal value of $(g_2, f_{21}, h_1)$. Note however
that this might be a locally optimal point. To come up with the global optimal
point, we may need to scan through the complete line $h_1$.}

\item $h_1 = \sqrt{\frac{\gamma \rho}{\rho + \sigma_w^2}}:$ In this case, we simply
solve (\ref{eq:f_21}) and (\ref{eq:g2}) to obtain the optimal values of
$f_{21}$ and $g_2$.
\end{itemize}

Having solved both the cases, the best post-processed SNR is given by the
maximum of the above two cases.

\subsection{Special Cases of Transmission with $N = 2$}
In this subsection, we look at the solution form of the general optimization
problem for special cases.

\subsubsection{Noiseless Feedback ($\sigma_n^2 = 0$)}
In this case,
  (\ref{eq:g2}) and (\ref{eq:f_21}) simplify to
\begin{equation*}
g_2 = \sqrt{\frac{\rho}{1 + \rho}} \textrm{ and } f_{21} = -\frac{\rho}{\sqrt{1
+ \rho}}.
\end{equation*}
Also, the optimal gain at the relay node is given by
\begin{equation}
h_1 = \min \left(\sqrt{\frac{\gamma \rho}{\rho + \sigma_w^2}} , \frac{1}{\sigma_w^2
\sqrt{1 + \rho}}\right).
\end{equation}

\subsubsection{Very Noisy Feedback ($\sigma_n^2 \rightarrow \infty$)}
With such a noisy feedback link, the value of side-information is drastically
reduced, implying $f_{21} = 0.$ Now from (\ref{c2}) and (\ref{SNR_g1_opt}), we obtain
\begin{equation*}
g_2 = \sqrt{\rho} \textrm{ and } h_1 = \min \left(\sqrt{\frac{\gamma \rho}{\rho + \sigma_w^2}},
\frac{1}{\sigma_w^2}\right).
\end{equation*}

\subsubsection{Noiseless Source-to-Relay Link ($\sigma_w^2 = 0$)}
For noiseless link, we use all the power available at the relay, i.e., $h_1 =
\sqrt{\gamma}$. Solving (\ref{eq:f_21}), we obtain the value of the optimal
$f_{21}$. This is followed by calculation of the optimal $g_2$ in (\ref{eq:g2})
yielding
\begin{equation}
g_2 = -\frac{1}{\sqrt{\rho}}\left(\frac{1 + \sigma_n^2(1 + \rho)}{1 + \sigma_n^2
f_{21}\sqrt{\gamma}}\right)f_{21}.
\end{equation}

\begin{figure*}
\begin{minipage}[b]{.48\linewidth}
  \centering
  \centerline{\epsfig{figure=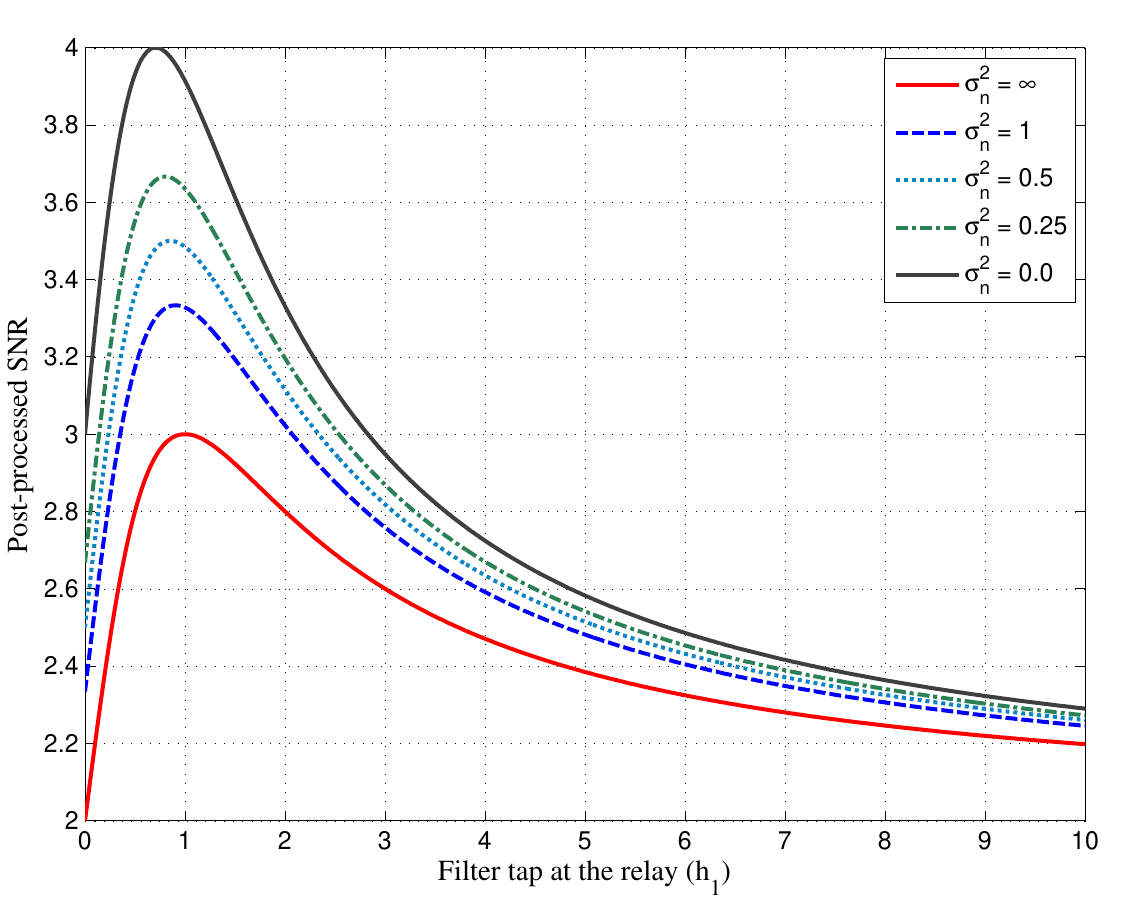,scale =
  0.45, angle = 0}}
  \centerline{\footnotesize (a)}\medskip
  \label{fig:noisy_1}
\end{minipage}
\hfill
\begin{minipage}[b]{0.48\linewidth}
  \centering
  \centerline{\epsfig{figure=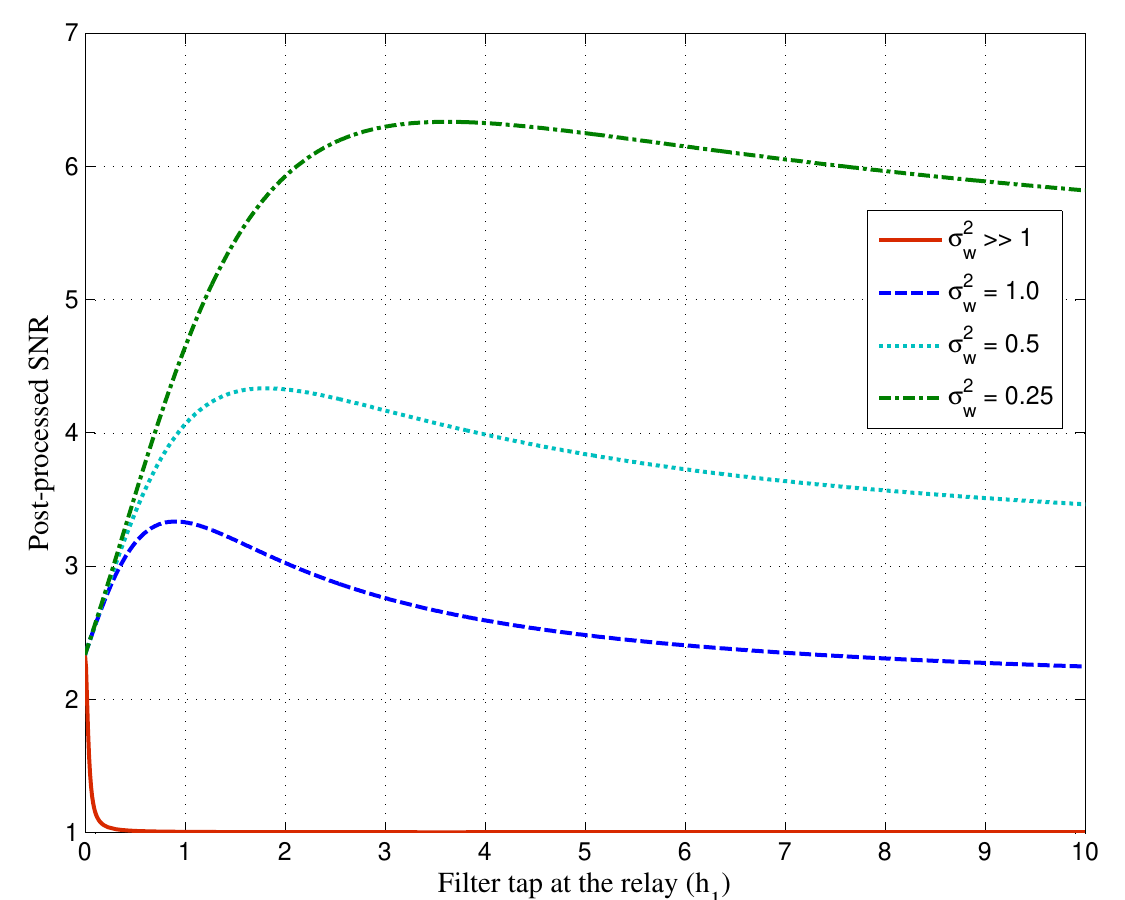,scale = 0.45, angle =
  0}} \vspace*{-0.2cm}
  \centerline{\footnotesize (b)}\medskip
  \label{fig:noisy_2}
\end{minipage}

\caption{(a) Plot of variation of post-processed SNR as a function of the noise
in the feedback link and the amplifier at the relay for $\rho = 1$ and
$\sigma_w^2 = 1.$ (b) Plot of variation of post-processed SNR as a function of the noise
in the source-to-relay link and the amplifier at the relay for $\rho = 1$ and
$\sigma_n^2 = 1$. }
\label{fig:noisy}
\end{figure*}
\subsubsection{Very Noisy Source-to-Relay Link ($\sigma_w^2 \rightarrow
\infty$)} In this case, the best strategy is to turn off the relay altogether, i.e., $h_1
= 0$. This then corresponds to the same scenario as analyzed in \cite{Bu69}. The
optimal parameters are now given by
\begin{align*}
g_2 & = \sqrt{\frac{\rho}{\left(1 + (1 + \rho)\sigma_n^2\right)^2 + \rho(1 +
\sigma_n^2)}}\left(\sigma_n^2 \rho + (1 + \sigma_n^2)\right) \\
f_{21} & = -\frac{\rho}{\sqrt{\left(1 + (1 + \rho)\sigma_n^2\right)^2 + \rho(1 +
\sigma_n^2)}}.\\
\end{align*}
%

Figures \ref{fig:noisy}(a) and \ref{fig:noisy}(b) plot the variation of the
post-processed SNR as a function of the gain at the relay $h_1$ for various feedback noise levels at
source-to-relay link and destination-source link. Figure \ref{fig:noisy}(a)
demonstrates that in the case of no output feedback at all, the
post-processed SNR is maximized at $h_1 = 1/\sigma_w^2 = 1$. With feedback, it
is seen that the maximum point shifts to the left which culminates in $h_1
= \frac{1}{\sigma_w^2\sqrt{1 + \rho}} = 0.707$ for the noiseless feedback case. We
see that improvements of up to $20\%$ is achievable even in the presence of
noisy feedback link.

Figure \ref{fig:noisy}(b) shows the impact on post-processed SNR as the noise
in the source-to-relay link is varied. For a very noisy source-to-relay link, it
is seen that the post-processed SNR is maximized by ignoring the relay node altogether.
As $\sigma_w^2$ decreases, the maximum is obtained {at a larger
filter tap coefficient} and in the limit that $\sigma_w^2 \rightarrow 0$, the
filter tap should be used to the maximum available power.

\section{Concluding Remarks}\label{sec:conc}
In this work we presented a lower bound on the capacity of the
three-terminal relay channel with the destination-to-source feedback. The bound
was obtained by drawing an equivalence between the three-terminal relay channel
and the single point-to-point communication link with feedforward noise
{having memory} of finite order. Using the recent results for
capacity of ARMA noise process in \cite{Ki10}, we derived improvements in the achievable rate for the relay
channel using very simple linear coding schemes at all the three terminals:
source, relay, and destination. While a tight lower bound for a general
ARMA$(p,q)$ effective noise process appears intractable, we demonstrated through numerical results the
advantages of using multiple taps at the relay node. We then extended the model to a
network of amplify-and-forward relays and proposed new lower bounds. 

We also explored the design of coding strategies that take noise in the
feedback link into consideration. {For the special case of two
transmissions}, $N = 2$, we proposed an optimal linear coding scheme that
maximizes the received SNR. This scheme can subsequently be used as an inner
code for a concatenated coding scheme that exploits channel output
feedback~\cite{ZaDa11}. Moreover, the successive refinement of the received
symbol at the destination automatically lends itself to be used in an ARQ
setting.

\bibliographystyle{IEEEtran}
\bibliography{Relay}
\begin{biography}[{\includegraphics[width=1in,height=1.25in,clip,keepaspectratio]{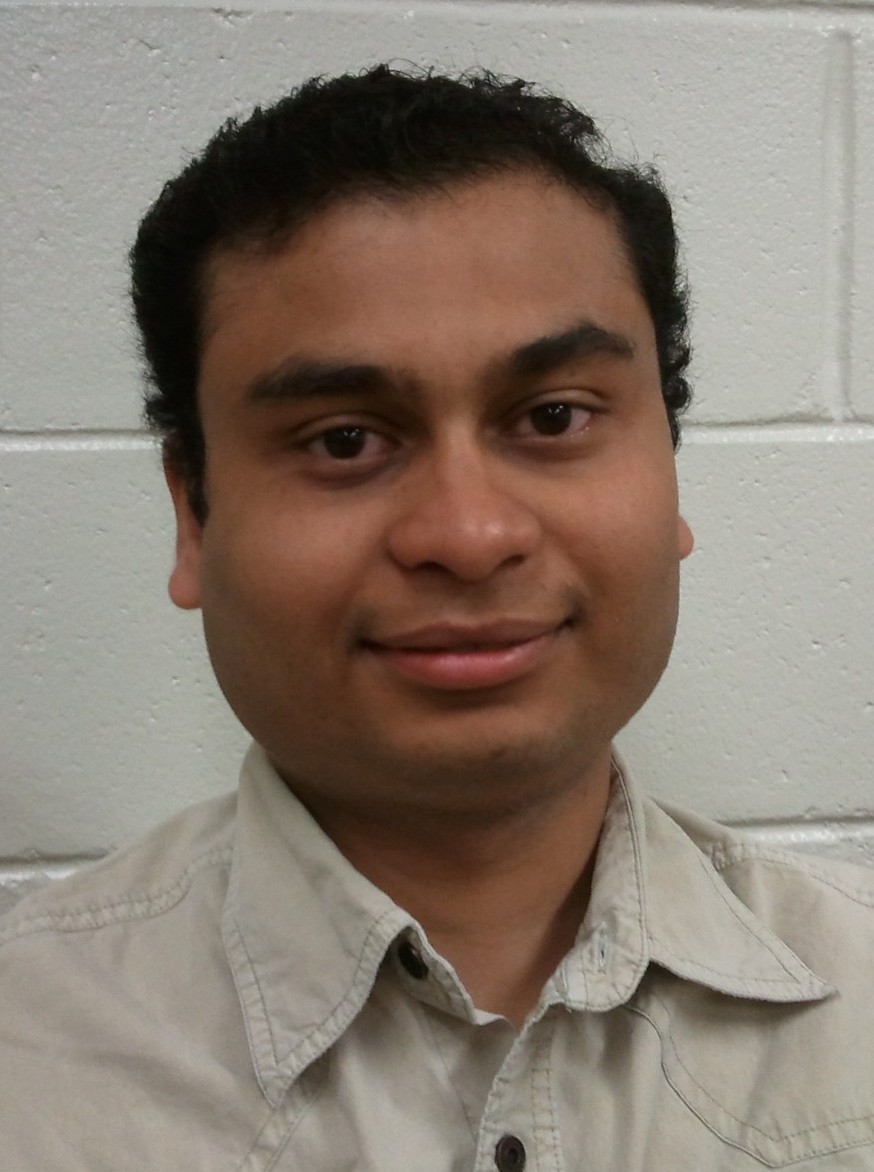}}]{Mayur
Agrawal} received the B. Tech degree in electronics and electrical communication
engineering in 2007 from Indian Institute of Technology Kharagpur, India.
Subsequently, in 2012, he completed his Ph.D. in electrical engineering at
Purdue University, West Lafayette, IN. Since November 2012, he has been with
WorldQuant Research (India) Private Limited. During the summers of 2008 and
2009, he was with Qualcomm Research Center, San Diego and Motorola iDEN Group,
Fort Lauderdale, respectively. He spent the summmer of 2011 as a Quantitative
Researcher with the Portfolio Construction Group at Citadel Investment Group,
Chicago. His research interests include the design of modern communication
systems, stochastic processes and time-series data. He has been a recipient of
the Magoon Award for Teaching Excellence in 2008 and 2009.
\end{biography}
\begin{biography}
[{\includegraphics[width=1in,height=1.25in,clip,keepaspectratio]{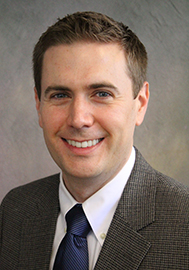}}]{David
J. Love} (S'98, M'05, SM'09, F'15)  received the B.S. (with highest honors),
M.S.E., and Ph.D. degrees in electrical engineering from the University of Texas
at Austin in 2000, 2002, and 2004, respectively. During the summers of 2000 and
2002, he was with Texas Instruments, Dallas, TX. Since August 2004, he has been
with the School of Electrical and Computer Engineering, Purdue University, West
Lafayette, IN, where he is now a Professor and recognized as a University
Faculty Scholar. He has served as an Editor for the IEEE Transactions on
Communications, an Associate Editor for the IEEE Transactions on Signal
Processing, and a guest editor for special issues of the IEEE Journal on
Selected Areas in Communications and the EURASIP Journal on Wireless
Communications and Networking. He is recognized as a Thomson Reuters Highly
Cited Researcher and holds 24 issued US patents.

Dr. Love is a Fellow of the Royal Statistical Society, and he has been inducted
into Tau Beta Pi and Eta Kappa Nu. Along with co-authors, he was awarded the
2009 IEEE Transactions on Vehicular Technology Jack Neubauer Memorial Award for
the best systems paper published in the IEEE Transactions on Vehicular
Technology in that year and multiple Globecom best paper awards. He was the
recipient of the Fall 2010 Purdue HKN Outstanding Teacher Award, Fall 2013
Purdue ECE Graduate Student Association Outstanding Faculty Award, and Spring
2015 Purdue HKN Outstanding Professor Award.
\end{biography}
\begin{biography}
[{\includegraphics[width=1in,height=1.25in,clip,keepaspectratio]{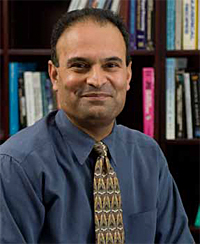}}]{
Venkataramanan Balakrishnan} (M'94, SM'06, F'12) received the B.Tech degree in
electronics and communication from the Indian Institute of Technology, Madras,
in 1985. He then attended Stanford University, where he received the M.S. degree
in statistics and the Ph.D. degree in electrical engineering in 1992.

Since 1994, Dr. Balakrishnan has served on the faculty of Electrical and
Computer Engineering at Purdue University, West Lafayette, Indiana, where he is
now Professor and Michael and Katherine Head. His primary research interests are
in applying numerical techniques, especially those based on convex optimization, to problems in
engineering. He is a co-author of the monograph Linear Matrix Inequalities in
System and Control Theory, published by SIAM, Philadelphia, in 1994.

Dr. Balakrishnan received the President of India Gold medal from the Indian
Institute of Technology, Madras, in 1985, the Young Investigator Award from the
Office of Naval Research in 1997, the Ruth and Joel Spira Outstanding Teacher
Award in 1998 and the Honeywell Award for excellence in teaching in 2001 from
the School of Electrical and Computer Engineering at Purdue University. He was
named a Purdue University Faculty Scholar in 2008.
\end{biography}
\end{document}